\documentclass[11pt]{article}
\usepackage[numbers,sort&compress]{natbib}
\usepackage[pagebackref]{hyperref}

\usepackage{tikz}
\usetikzlibrary{arrows.meta, positioning}

\usepackage{amsfonts}

\usepackage{amssymb}
\usepackage{amsthm}
\usepackage{amsmath}
\usepackage[margin=1in]{geometry}

\usepackage[appendix=inline]{apxproof}

\RenewEnviron{appendixproof}[1][]{\begin{proof} \BODY \end{proof}}

\newcommand{\CCfont}[1]{\ensuremath{\mathsf{#1}}}
\newcommand{\dimfont}[1]{\ensuremath{\mathsf{#1}}}

\newcommand{\MCSP}{\CCfont{MCSP}}

\newcommand{\pair}[1]{\langle #1 \rangle}
\newcommand{\ceil}[1]{ \left\lceil #1 \right\rceil }
\newcommand{\myset}[2]{ \left\{ #1 \left| #2 \right.\right\} }
\newcommand{\prefix}{\sqsubseteq}
\newcommand{\liminfn}{\liminf\limits_{n\to\infty}}
\newcommand{\limsupn}{\limsup\limits_{n\to\infty}}
\newcommand{\PR}[2]{\underset{#1}{\CCfont{Pr}}\left[#2\right]}
\newcommand{\restr}{\negthinspace\upharpoonright\negthinspace}
\newcommand{\parityP}{{\oplus\P}}
\newcommand{\Sharp}{\CCfont{\#}}
\newcommand{\sharpP}{{\Sharp\P}}
\newcommand{\ParityP}{\parityP}
\newcommand{\SharpP}{\sharpP}

\newcommand{\N}{\mathbb{N}}

\newcommand{\p}{{\CCfont{p}}}
\newcommand{\pspace}{{\CCfont{pspace}}}
\newcommand{\comp}{{\CCfont{comp}}}
\newcommand{\mup}{\mu_\p}
\newcommand{\mupspace}{\mu_\pspace}
\renewcommand{\dim}{{\dimfont{dim}}}
\newcommand{\dimp}{\dim_\p}
\newcommand{\dimpspace}{\dim_\pspace}
\newcommand{\dimh}{\dim_\dimfont{H}}
\newcommand{\dimcomp}{\dim_\comp}
\newcommand{\cdim}{\dimfont{cdim}}
\newcommand{\Deltapthree}{{\Delta^\p_3}}
\newcommand{\dimDeltapthree}{\dim_\Deltapthree}
\newcommand{\Deltaptwo}{{\Delta^\p_2}}
\newcommand{\DeltaPtwo}{{\Delta^\P_2}}
\newcommand{\DeltaPthree}{{\Delta^\P_3}}
\newcommand{\dimpack}{{\dim}_{\dimfont{P}}}
\newcommand{\Dim}{{\dimfont{Dim}}}
\newcommand{\cDim}{{\dimfont{cDim}}}
\newcommand{\Dimp}{{\Dim_\p}}
\newcommand{\Dimcomp}{{\Dim_\comp}}
\newcommand{\Dimpspace}{{\Dim_\pspace}}
\newcommand{\str}{{\CCfont{str}}}

\newcommand{\C}{\CCfont{C}}

\newcommand{\co}[1]{\CCfont{co}#1}
\newcommand{\strings}{\{0,1\}^*}
\newcommand{\ALL}{\CCfont{ALL}}
\newcommand{\CE}{\CCfont{CE}}
\newcommand{\SPARSE}{\CCfont{SPARSE}}
\newcommand{\DENSE}{\CCfont{DENSE}}
\newcommand{\DTIME}{\CCfont{DTIME}}
\newcommand{\BPTIME}{\CCfont{BPTIME}}
\newcommand{\NP}{\CCfont{NP}}
\newcommand{\FewP}{\CCfont{FewP}}
\newcommand{\UP}{\CCfont{UP}}
\newcommand{\SIZE}{\CCfont{SIZE}}
\newcommand{\BQSIZE}{\CCfont{BQSIZE}}
\newcommand{\SIZEio}{\SIZE_\CCfont{i.o.}}
\newcommand{\BQP}{\CCfont{BQP}}
\newcommand{\BPP}{\CCfont{BPP}}
\newcommand{\E}{\CCfont{E}}
\newcommand{\NE}{\CCfont{NE}}
\newcommand{\coNE}{\co\NE}
\newcommand{\UE}{\CCfont{UE}}
\newcommand{\coUE}{\co\UE}
\newcommand{\BPE}{\CCfont{BPE}}
\newcommand{\PSPACE}{\CCfont{PSPACE}}
\newcommand{\ESPACE}{\CCfont{ESPACE}}
\newcommand{\DEC}{\CCfont{DEC}}

\renewcommand{\P}{\CCfont{P}}
\newcommand{\EXP}{\CCfont{EXP}}
\newcommand{\DeltaE}{\Delta^\E}
\newcommand{\DeltaEtwo}{\DeltaE_2}
\newcommand{\DeltaEthree}{\DeltaE_3}
\newcommand{\SigmaP}{\Sigma^\P}
\newcommand{\SigmaPtwo}{{\SigmaP_2}}
\newcommand{\PH}{\CCfont{PH}}

\newcommand{\poly}{\CCfont{poly}}
\newcommand{\Ppoly}{{\P/\poly}}
\newcommand{\qpoly}{\CCfont{qpoly}}
\newcommand{\PP}{\CCfont{PP}}
\newcommand{\SPP}{\CCfont{SPP}}
\newcommand{\AWPP}{\CCfont{AWPP}}
\newcommand{\GapP}{\CCfont{GapP}}

\newcommand{\FP}{\CCfont{FP}}

\newcommand{\calC}{{\mathcal C}}
\newcommand{\calD}{{\mathcal D}}
\newcommand{\calH}{{\mathcal H}}
\newcommand{\calM}{{\mathcal M}}

\newcommand{\calKS}{{\mathcal KS}}
\newcommand{\calK}{{\mathcal K}}

\newcommand{\io}{\CCfont{i.o.}}
\newcommand{\aev}{{\CCfont{a.e.}}}
\newcommand{\ioA}{A^\io}
\newcommand{\aeA}{A^\aev}

\newcommand{\HP}{\calH_\P}
\newcommand{\HPSPACE}{\calH_\PSPACE}
\newcommand{\HNP}{\calH_\NP}
\newcommand{\HC}{\calH_\calC}
\newcommand{\spanP}{\CCfont{SpanP}}
\newcommand{\SpanP}{\spanP}
\renewcommand{\HC}{\mathrm{HC}}

\newcommand{\nth}{n^{\mathrm{th}}}

\renewcommand{\SIZEio}{\SIZE^\CCfont{i.o.}}

\newcommand{\T}{\CCfont{T}}
\newcommand{\PT}{\CCfont{P}_{\CCfont{T}}}
\renewcommand{\Pr}{\CCfont{P}_r}

\newcommand{\PnaT}{\CCfont{P}_{n^\alpha\CCfont{-T}}}

\newcommand{\PTM}{\CCfont{PTM}}
\newcommand{\QTM}{\CCfont{QTM}}

\newcommand{\musharpp}{\mu_{\SharpP}}
\newcommand{\mugapp}{\mu_{\GapP}}
\newcommand{\muspanp}{\mu_{\SpanP}}
\newcommand{\muSpanP}{\muspanp}

\newcommand{\dimsharpp}{\dim_{\SharpP}}
\newcommand{\dimgapp}{\dim_{\GapP}}
\newcommand{\dimspanp}{\dim_{\SpanP}}
\newcommand{\dimSpanP}{\dimspanp}
\newcommand{\dimGapP}{\dim_\GapP}
\newcommand{\dimSharpP}{\dim_\SharpP}

\newcommand{\PSharpP}{{\P^\SharpP}}
\newcommand{\muPSharpP}{\mu_\PSharpP}

\newcommand{\Dimsharpp}{\Dim_{\SharpP}}
\newcommand{\Dimspanp}{\Dim_{\SpanP}}
\newcommand{\DimSpanP}{\Dimspanp}
\newcommand{\DimGapP}{\Dim_\GapP}
\newcommand{\DimSharpP}{\Dim_\SharpP}

\newcommand{\muSharpP}{\mu_\SharpP}
\newcommand{\muGapP}{\mu_\GapP}

\newcommand{\dimDeltaPthree}{\dim_\DeltaPthree}
\newcommand{\muDeltaPthree}{\mu_\DeltaPthree}

\def\binary{\lbrace 0,1 \rbrace}
\def\Pr{\CCfont{Pr}}

\newcommand{\NPSV}{\CCfont{NPSV}}
\newcommand{\UPSV}{\CCfont{UPSV}}
\newcommand{\Few}{\CCfont{Few}}

\renewcommand{\FP}{\CCfont{FP}}

\newcommand{\ClassMeasure}[1]{{\calM_{#1}}}
\newcommand{\MC}{\ClassMeasure{\calC}}
\newcommand{\MNP}{\ClassMeasure{\NP}}

\renewcommand{\HC}{\calH_\calC}

\newcommand{\lcm}{\CCfont{lcm}}

\newcommand{\twonn}{\frac{2^n}{n}}
\newcommand{\lutzbound}{\twonn\left(1+\frac{\alpha \log n}{n}\right)}
\newcommand{\dimbound}{\alpha\twonn}
\newcommand{\SIZEiolutzbound}{\SIZEio\!\left(\lutzbound\right)}
\newcommand{\SIZEdimbound}{\SIZE\!\left(\dimbound\right)}
\newcommand{\SIZEiodimbound}{\SIZEio\!\left(\alpha \frac{2^n}{n}\right)}

\newcommand{\BQTIME}{\CCfont{BQTIME}}
\newcommand{\BQE}{\CCfont{BQE}}

\newcommand{\FSPtwo}{\CCfont{FS}^\P_2}
\newcommand{\SEtwo}{\CCfont{S}^\E_2}

\usepackage{amsthm}

\newtheorem{theorem}{Theorem}[section]
\newtheorem{corollary}[theorem]{Corollary}
\newtheorem{lemma}[theorem]{Lemma}
\newtheorem{proposition}[theorem]{Proposition}

\newtheorem{conjecture}[theorem]{Conjecture}

\newenvironment{theorem_cite}[1]
{\begin{theorem}  {\rm (#1)}}
{\end{theorem}}
\newenvironment{lemma_cite}[1]
{\begin{lemma}  {\rm (#1)}}
{\end{lemma}}

\newenvironment{lemma_named}[1]
{\begin{lemma} {\rm (#1)}}
{\end{lemma}}
\newenvironment{definition_named}[1]
{\begin{definition} {\rm (#1)}}
{\end{definition}}
\newenvironment{construction_named}[1]
{\begin{construction}  {\rm (#1)}}
{\end{construction}}

\newtheorem{question}[theorem]{Question}

\newtheorem*{hypothesis}{Hypothesis}

\theoremstyle{definition}

\newtheorem*{definition}{Definition}
\newtheorem*{defn}{Definition}

\newenvironment{definition_cite}[1]
{\begin{definition}  {\rm (#1)}}
{\end{definition}}

\newtheorem{example}[theorem]{Example}

\newtheorem*{example*}{Example}
\newtheorem*{examples*}{Examples}

\newtheorem{construction}[theorem]{Construction}

\theoremstyle{remark}

\numberwithin{equation}{section}
\numberwithin{figure}{section}

\usepackage{paralist}
\newenvironment{enumerateC}{\begin{enumerate}}{\end{enumerate}}
 \newcommand{\SPE}{\CCfont{SPE}}
\newcommand{\GapE}{\CCfont{GapE}}
\newcommand{\AWPE}{\CCfont{AWPE}}

\def\binary{\lbrace 0,1 \rbrace}

\renewcommand{\C}{\CCfont{C}}
\renewcommand{\SIZE}{\CCfont{SIZE}}
\renewcommand{\SIZEio}{\CCfont{SIZE}^\CCfont{i.o.}}
\renewcommand{\BQSIZE}{\CCfont{BQSIZE}}
\newcommand{\BQSIZEio}{\BQSIZE^\CCfont{i.o.}}

\renewcommand{\lutzbound}{\twonn\!\left(1+\frac{\alpha \log n}{n}\right)}

\renewcommand{\dimpack}{\Dim_\mathsf{pack}}

\renewcommand{\p}{\P}
\renewcommand{\pspace}{\PSPACE}
 \newtheorem*{acks}{Acknowledgments}

\title{\bf Counting Martingales for\\Measure and Dimension in Complexity Classes}
\date{}

\usepackage{ifthen}

\newboolean{showIgnored}
\setboolean{showIgnored}{false}
\newboolean{showComments}
\setboolean{showComments}{true}
\newcommand{\sectionnewpage}{\newpage}
\renewcommand{\sectionnewpage}{}

\author{John M. Hitchcock
\thanks{Department of Electrical Engineering and Computer Science, University of Wyoming. jhitchco@uwyo.edu. This research was supported in part by NSF grant 2431657.}
\and Adewale Sekoni
\thanks{Department of Mathematics, Computer Science \& Physics,
  Roanoke College. sekoni@roanoke.edu.}
\and
Hadi Shafei
\thanks{School of Computing and Engineering, University of Huddersfield. H.Shafei@hud.ac.uk.}
}

\begin{document}
\maketitle

\begin{abstract}
	This paper makes two primary contributions. First, we introduce the concept of {\em counting martingales} and use it to define {\em counting measures}, {\em counting dimensions}, and {\em counting strong dimensions}. Second, we apply these new tools to strengthen previous {\em circuit lower bounds}.

	Resource-bounded measure and dimension have traditionally focused on deterministic time and space bounds. We  use counting complexity classes to develop resource-bounded counting measures and dimensions. Counting martingales are constructed using functions from the $\SharpP$, $\SpanP$, and $\GapP$ complexity classes. We show that counting martingales capture many
	martingale constructions in complexity theory. The resulting counting measures and dimensions are intermediate in power between the standard time-bounded and space-bounded notions, enabling finer-grained analysis where space-bounded measures are known, but time-bounded measures remain open.
	For example, we show that $\BPP$ has $\sharpP$-dimension 0 and $\BQP$ has $\GapP$-dimension 0, whereas the $\P$-dimensions of these classes remain open.

	As our main application, we improve circuit-size lower bounds. Lutz (1992) strengthened Shannon's classic $(1-\epsilon)\twonn$ lower bound (1949) to $\pspace$-measure, showing that almost all problems require circuits of size $\lutzbound$, for any $\alpha < 1$. We extend this result to $\SpanP$-measure, with a proof that uses a connection through the Minimum Circuit Size Problem ($\MCSP$) to construct a counting martingale. Our results imply that the stronger lower bound holds within the third level of the exponential-time hierarchy, whereas previously, it was only known in $\ESPACE$. Under a derandomization hypothesis, this lower bound holds within the second level of the exponential-time hierarchy, specifically in the class $\E^\NP$. We study the $\SharpP$-dimension of classical circuit complexity classes and the $\GapP$-dimension of quantum circuit complexity classes. We also show that if one-way functions exist, then $\SharpP$-dimension is strictly more powerful than $\P$-dimension.
\end{abstract}

\newpage
\tableofcontents
\listoffigures
\newpage

\sectionnewpage
\section{Introduction}\label{sec:intro}

\hyphenation{
	re-source-bound-ed mea-sure
	compu-tational complex-ity theo-ry
	poly-nomial-time mea-sure
	poly-nomial-space mea-sure
}

Resource-bounded measures and dimensions \cite{Lutz:thesis,Lutz:CMCC,Lutz:AEHNC,Lutz:DCC,Athreya:ESDAICC} are fundamental tools for analyzing complexity classes, offering refined ways to understand the power and limitations of different computational resources \cite{Lutz:QSET,Hitchcock:FGCC,Lutz:TPRBM,AmbMay97}.
Traditional approaches based on real-valued martingales provide insights into classes like $\P$, $\NP$, $\PSPACE$, and $\EXP$, but they leave gaps when it comes to many intermediate complexity classes. This paper introduces counting martingales, which offer a new perspective on these intermediate classes by utilizing functions from counting complexity classes.

Our main contributions are:
\begin{enumerate}
	\item Counting Martingales and Counting Measures and Dimensions: We introduce {\em counting martingales}, which generalize traditional martingales by incorporating functions from
	      counting complexity classes \cite{Valiant79:Permanent,KoScTo89,FeFoKu94,Li93}, creating intermediate {\em counting measures and dimensions}.  For instance, we define $\SharpP$-measure, $\SpanP$-measure, and $\GapP$-measure as intermediate measures between $\P$-measure and $\PSPACE$-measure, allowing finer analysis of complexity classes.
	\item Applications to Circuit Complexity: Using these new measures, we provide novel results on nonuniform complexity and the Shannon-Lupanov bound \cite{Shannon49,Lupa58}. Shannon's  $(1-\epsilon)\twonn$ lower bound was improved by Lutz \cite{Lutz:AEHNC} who showed that for any $\alpha < 1$, almost all problems require circuits of size $\lutzbound$. We improve this result further by extending it to $\SpanP$-measure. In our proof, we construct a counting martingale using a connection through the Minimum Circuit Size Problem ($\MCSP$). Moreover, we study classical  and quantum circuit complexity classes using $\SharpP$-dimension and $\GapP$-dimension, respectively.
\end{enumerate}

The standard definitions of resource-bounded measure use martingales that are real-valued functions computable within deterministic time and space resource bounds \cite{Lutz:AEHNC}.
In this paper, we use the counting complexity classes $\SharpP$ \cite{Valiant79:Permanent}, $\SpanP$ \cite{KoScTo89}, and $\GapP$ \cite{FeFoKu94,Li93}
to introduce the concept of {\em counting martingales} and define {\em counting measures}, {\em counting dimensions}, and {\em counting strong dimensions} that are intermediate in power between the previous time- and space-bounded measures and dimensions. A $\SharpP$ function counts the number of accepting paths of a probabilistic Turing machine ($\PTM$), while a $\SpanP$ function counts the number of different outputs of a $\PTM$. Every $\SharpP$ function is also a $\SpanP$ function, and the classes are equal if and only if $\UP = \NP$ \cite{KoScTo89}.
A $\GapP$ function counts the difference between the number of accepting paths and the number of rejecting paths of a $\PTM$ \cite{FeFoKu94,Li93}. For more background on counting complexity, we refer to the surveys \cite{Schoning90,Fort97}.

A  martingale is
a function $d : \{0,1\}^* \to [0,\infty)$ such that \[d(w) =
	\frac{d(w0)+d(w1)}{2}\] for all $w \in \{0,1\}^*$.
We view a martingale as acting on Cantor Space $\C = \{0,1\}^\infty$ (the infinite binary tree). The value at any node is the average of the values below it. By induction, the value at any node is also the average of the values at any level of the subtree below the node: \[d(w) = \frac{1}{2^n} \sum_{x \in \{0,1\}^n} d(wx).\]
A martingale starts with a finite amount $d(\lambda)$ at the root. We may assume $d(\lambda) = 1$ without loss of generality.

Intuitively, because the average value of a martingale across $\{0,1\}^n$ is $d(\lambda) = 1$, a martingale is unable to obtain ``large'' values on ``many'' sequences. This intuition is formalized into characterizations of
Lebesgue measure \cite{Lebesgue1902}, Hausdorff dimension \cite{Haus19}, and packing dimension \cite{Athreya:ESDAICC} using martingales. A class $X \subseteq \C$ has measure 0 if and only if there is a martingale that attains unbounded values on all elements of $X$ \cite{Vill39}. A class $X \subseteq \C$ has Hausdorff dimension $s$ if $2^{(1-s)n}$ is the optimal infinitely-often growth rate of martingales on $X$ \cite{Lutz:DCC}.
A class $X \subseteq \C$ has packing dimension $s$ if $2^{(1-s)n}$ is the optimal almost-everywhere growth rate of martingales on $X$ \cite{Athreya:ESDAICC}.

Resource-bounded measure and dimension come from restricting these characterizations to martingales computable within some complexity class \cite{Lutz:AEHNC,Lutz:DCC,Athreya:ESDAICC}. The most used resource bounds are polynomial-time $(\p)$ and polynomial-space $(\pspace)$.
Generally, it is much easier to construct $\pspace$-martingales than it is $\p$-martingales \cite{Lutz:CWCPSC,Hitchcock:SDKCTHS}. For more background on resource-bounded measure and dimension we refer to the surveys \cite{Lutz:QSET,Lutz:TPRBM,AmbMay97,Hitchcock:FGCC,Lutz:EFD,Stull20,Mayordomo:EFDAIT,Mayordomo:EHD}.

\subsection{Counting Martingales}\label{subsec:intro_counting_martingales}
\begin{figure}
	\begin{center}
		\begin{tikzpicture}[level distance=1.5cm,
				level 1/.style={sibling distance=12cm},
				level 2/.style={sibling distance=6cm},
				level 3/.style={sibling distance=3cm},
				level 4/.style={sibling distance=1.5cm},scale=.47,
				every node/.style={draw, rectangle, align=center}]
			\node {$\frac{5}{16}$}
			child {node {$\frac{1}{2}$}
					child {node {$\tfrac{3}{4}$}
							child {node {$\tfrac{1}{2}$}
									child {node {0}}
									child {node [fill=green] {1}}
								}
							child {node {1}
									child {node [fill=green] {1}}
									child {node [fill=green] {1}}
								}
						}
					child {node {$\tfrac{1}{4}$}
							child {node {0}
									child {node {0}}
									child {node {0}}
								}
							child {node {$\tfrac{1}{2}$}
									child {node [fill=green] {1}}
									child {node {0}}
								}
						}
				}
			child {node {$\frac{1}{8}$}
					child {node {0}
							child {node {0}
									child {node {0}}
									child {node {0}}
								}
							child {node {0}
									child {node {0}}
									child {node {0}}
								}
						}
					child {node {$\frac{1}{4}$}
							child {node {$\frac{1}{2}$}
									child {node {0}}
									child {node [fill=green] {1}}
								}
							child {node {0}
									child {node {0}}
									child {node {0}}
								}
						}
				};
		\end{tikzpicture}
	\end{center}\caption{Cover Martingale Construction (Construction \ref{construction:cover_martingale})}\label{figure:martingale}
\end{figure}
Our main conceptual contribution is the introduction of {\em counting martingales}, which provide intermediate measure and dimension notions between $\P$ and $\PSPACE$ by using functions from counting complexity classes to define martingales. Many martingale constructions in complexity theory are expressed naturally as counting martingales.

For example, a standard technique for constructing a martingale is betting on a
cover (see Figure \ref{figure:martingale}) \cite{Lutz:QSET,AmbMay97}.
Given a set $A \subseteq \{0,1\}^*$ and a length $n \geq 0$, we choose $x \in \{0,1\}^n$ uniformly at random and define a martingale
$d_A$ by
\[d_A(w) = \PR{x \in \{0,1\}^n}{x \in A \mid w \prefix x} =
	\frac{| \{ x \in A \cap \{0,1\}^n \mid w \prefix x \} |}{2^{n - |w|}}\]
for all $w \in \{0,1\}^{\leq n}$. Strings of longer length have the value of their length-$n$ prefix. If we can construct an infinite family of such martingales that cover a class and the sum of their initial values converges, the Borel-Cantelli lemma applies to show the class has measure 0 \cite{DowHir:book}.

A key difference between space-bounded measure and time-bounded
measure is that a space-bounded martingale can enumerate a covering,
whereas a time-bounded martingale does not have time to do this.
The complexity of computing $d_A$ depends on the complexity of $A$ and the enumeration bottleneck.
Suppose that $A \in \P$.
Naively computing $d_A$ would involve an
enumeration of $A$, requiring polynomial space.
Computing $d_A$ in polynomial time would only be possible if we have some special
structure in $A$. For example, if $A$ is $\P$-rankable \cite{AllRub88}, then $d_A$ is
polynomial-time computable \cite{Hitchcock:DERC}. It is also possible
to approximately compute $d_A$ using an oracle from the
polynomial-time hierarchy \cite{Stoc85,Mayo94b,Hitchcock:DERC}.

The numerator in the definition of $d_A$ is in general a $\SharpP$
function if $A \in \P$. This is because we can use a $\PTM$ (Probabilistic Turing Machine) to count how many extensions of a string are in the cover. We call this a {\em $\SharpP$-martingale}. In general, a martingale is a $\SharpP$-martingale if it can be approximated by the
ratio of a $\SharpP$ function and a polynomial-time function.

The case when the cover $A \in \NP$ is also interesting, for example when $A$ is the Minimum Circuit Size Problem ($\MCSP$). Then the numerator is a $\SpanP$ function.
A martingale of this form is a {\em $\SpanP$-martingale.}

When we use a $\GapP$ function for the numerator of a martingale, we call it a {\em $\GapP$-martingale}. We will show that $\GapP$-martingales are capable of measuring quantum complexity classes. Until now, quantum complexity has not been addressed by resource-bounded measure and dimension.

We call $\SharpP$-martingales, $\SpanP$-martingales, and $\GapP$-martingales {\em counting martingales}.
Definition \ref{def:counting_martingales} contains the formal definitions of counting martingales.

\subsection{Counting Measures and Counting Dimensions}

A class $X$ has $\SharpP$-measure 0, written $\mu_{\SharpP}(X) = 0$,
if there is a $\SharpP$-martingale that succeeds on $X$.
Analogously, a class $X$ has $\SpanP$-measure 0, written $\mu_{\SpanP}(X) = 0$,
if there is a $\SpanP$-martingale that succeeds on $X$. Furthermore, a class $X$ has
$\GapP$-measure 0, written $\mu_{\GapP}(X) = 0$, if there is a $\GapP$-martingale that succeeds on $X$.
These {\em counting measures} are
intermediate between $\p$-measure and $\pspace$-measure \cite{Lutz:AEHNC}:
for every class $X \subseteq \C$,
\[\begin{array}{ccccccc}
		\mup(X) = 0 & \Rightarrow & \musharpp(X)=0  & \Rightarrow & \mugapp(X) = 0    \\
		            &             & \Downarrow      &             & \Downarrow        \\
		            &             & \muspanp(X) = 0 & \Rightarrow & \mupspace(X) = 0.\end{array}
\]We do not know of any relationship between $\muSpanP$ and $\muGapP$.
An individual problem $B$ is {\em $\Delta$-random} if no $\Delta$-martingale succeeds on $B$.
We show that $\UE$ problems are not $\SharpP$-random and $\NE$ problems are not $\SpanP$-random, where $\UE$ and $\NE$ are the exponential-time versions of $\UP$ and $\NP$.
When we have a proposition that is known in $\pspace$-measure, but open in
$\p$-measure, we can investigate it in $\SharpP$-measure, $\SpanP$-measure, or $\GapP$-measure.

We also introduce {\em counting dimensions}, $\SharpP$-dimension, $\SpanP$-dimension, and $\GapP$-dimension, written $\dimSharpP(X)$, $\dimSpanP(X)$, and $\dimGapP(X)$, respectively.
These dimensions analogously fall between $\P$-dimension and $\PSPACE$-dimension \cite{Lutz:DCC}. For all $X \subseteq \C$,
\newcommand{\rotleq}{\mathrel{\raisebox{1ex}{\rotatebox{270}{$\leq$}}}}
\[\begin{array}{ccccccc}
		0 & \leq & \dimpspace(X) & \leq & \dimGapP(X)                             \\
		  &      & \rotleq       &      & \rotleq                                 \\
		  &      & \dimSpanP(X)  & \leq & \dimSharpP(X) & \leq & \dimp(X) \leq 1.
	\end{array}
\]We do not know of any relationship between $\dimSpanP$ and $\dimGapP$.
We also develop {\em counting strong dimensions}, written $\DimSharpP(X)$, $\DimSpanP(X)$, and $\DimGapP(X)$, respectively.
These dimensions similarly fall between $\P$-strong dimension and $\PSPACE$-strong dimension \cite{Athreya:ESDAICC}.
For all $X \subseteq \C$,
\[\begin{array}{ccccccc}
		0
		 & \leq & \Dimpspace(X) & \leq & \DimGapP(X)                             \\
		 &      & \rotleq       &      & \rotleq                                 \\
		 &      & \DimSpanP(X)  & \leq & \DimSharpP(X) & \leq & \Dimp(X) \leq 1.
	\end{array}\]
Strong dimension is a more stringent criterion, requiring success to hold for almost all input lengths, rather than just for infinitely many. We work out the definitions and basic properties of counting measures and dimensions in Section \ref{sec:counting_martingales}.

\subsection{Our Techniques}

Many martingale constructions in complexity theory are expressed naturally as counting martingales.
These and other martingale constructions in this paper follow similar patterns.
In Section \ref{sec:constructions} we present five martingale constructions in a unified framework.

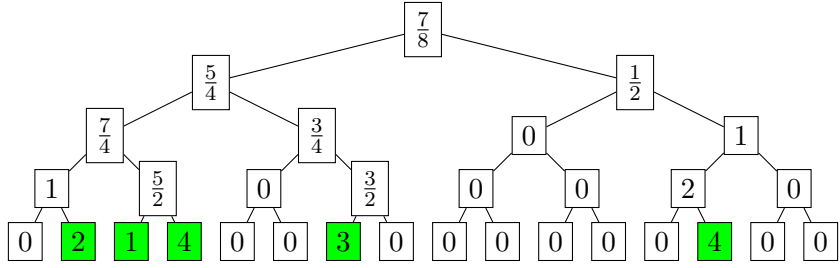
\begin{figure}
	\begin{center}
		\begin{tikzpicture}[level distance=1.5cm,
				level 1/.style={sibling distance=12cm},
				level 2/.style={sibling distance=6cm},
				level 3/.style={sibling distance=3cm},
				level 4/.style={sibling distance=1.5cm},scale=.47,
				every node/.style={draw, rectangle, align=center}]
			\node {$\tfrac{7}{8}$}
			child {node {$\frac{5}{4}$}
			child {node {$\tfrac{7}{4}$}
			child {node {1}child {node {0}}
			child {node [fill=green] {2}}
			}
			child {node {$\tfrac{5}{2}$}
					child {node [fill=green] {1}}
					child {node [fill=green] {4}}
				}
			}
			child {node {$\tfrac{3}{4}$}
					child {node {0}
							child {node {0}}
							child {node {0}}
						}
					child {node {$\tfrac{3}{2}$}
							child {node [fill=green] {3}}
							child {node {0}}
						}
				}
			}
			child {node {$\frac{1}{2}$}
					child {node {0}
							child {node {0}
									child {node {0}}
									child {node {0}}
								}
							child {node {0}
									child {node {0}}
									child {node {0}}
								}
						}
					child {node {$1$}
							child {node {$2$}
									child {node {0}}
									child {node [fill=green] {4}}
								}
							child {node {0}
									child {node {0}}
									child {node {0}}
								}
						}
				};
		\end{tikzpicture}
	\end{center}\caption{Conditional Expectation Martingale Construction (Construction \ref{construction:random_variable_martingale})}\label{figure:martingale_random_variable}
\end{figure}
\begin{enumerate}
	\item \textbf{Cover Martingale.} The Cover Martingale construction utilizes a cover set \( A \) and defines the martingale based on the conditional probability that an extension of the current string belongs to \( A \). If \( A \in \UP \), this construction produces a \(\SharpP\)-martingale; if \( A \in \NP \), it results in a \(\SpanP\)-martingale. If $A$ is in the counting complexity class $\SPP$, then this generates a $\GapP$-martingale. The class  $\SPP$ consists of all languages $L$ for which  there exists a $\GapP$ function $f$ such that if $x \in L$ then $f(x) = 1$, and $f(x) = 0$ otherwise \cite{FeFoKu94}. This approach is effective for capturing known martingale constructions in the literature, particularly in scenarios where membership within a subset can be determined with unique or nondeterministic witnesses. For further details, refer to Figure \ref{figure:martingale} and Construction \ref{construction:cover_martingale}.

	\item \textbf{Conditional Expectation Martingale.} This martingale generalizes the Cover Martingale by using a counting function \( f(x) \) as a random variable and taking its conditional expectation given the current prefix \( w \). This construction is adaptable to functions within \(\SharpP\), \(\SpanP\) and $\GapP$, making it effective for applications that require combining information from extensions of a given string. The Conditional Expectation Martingale averages the values of \( f \) over all extensions, providing a flexible tool. See Figure \ref{figure:martingale_random_variable} and Construction \ref{construction:random_variable_martingale} for an example and technical details.

	\item \textbf{Subset Martingale.} This martingale is designed to succeed on all infinite subsets of a language \( B \). If \( B \in \UP \), it produces a \(\SharpP\)-martingale, and if \( B \in \NP \), it yields a \(\SpanP\)-martingale. If $B \in \SPP$, then this is a $\GapP$-martingale. Unlike deterministic martingales that can focus directly on a language, the Subset Martingale is unique in its ability to succeed across all subsets of a language. This property allows us to conclude that \(\SharpP\)-random languages are not in \(\UE\), \(\SpanP\)-random languages are not in \(\NE\), and $\GapP$-random languages are not in $\SPE$ (the exponential version of $\SPP$). See Figure \ref{figure:subset_martingale} for an example and Construction \ref{construction:subset_martingale} for a formal definition.

	\item \textbf{Acceptance Probability Martingale.}
	      We also show that betting according to the acceptance probabilities of a $\PTM$ or $\QTM$ (quantum Turing machine) yields a counting martingale. Using this, we show that $\BPP$ has $\SharpP$-dimension 0 and $\BQP$  has $\GapP$-dimension 0. See Figure \ref{figure:acceptance_probability_martingale} for an example and Construction \ref{construction:probability_martingale_construction} for a formal definition.

	\item  \textbf{Bi-Immunity Martingale.} Mayordomo \cite{Mayo94} showed $\P$-random languages are $\E$-bi-immune and $\PSPACE$-random languages are $\ESPACE$-bi-immune using a construction that we call a bi-immunity martingale.
	      This construction is designed to succeed on all supersets of an infinite language. We show that this construction works as counting martingales to show $\SharpP$-random languages are $\UE \cap \coUE$-bi-immune, $\SpanP$-random languages are $\NE \cap \coNE$-bi-immune, and $\GapP$-random languages are $\SPE$-bi-immune. See Figure \ref{figure:biimunity_martingale} for an example and Construction \ref{construction:bi-immunity_martingale} for a formal definition.
\end{enumerate}

\noindent
{\bf Entropy Rates.} Hitchcock and Vinodchandran \cite{Hitchcock:DERC} showed a covering notion called the $\NP$-entropy rate is an upper bound for $\Deltapthree$-dimension, where $\Deltapthree = \P^{\SigmaPtwo}$. In Section \ref{subsec:entropy_rates}, we extend this by showing that $\SpanP$-dimension lies between $\Deltapthree$-dimension and the $\NP$-entropy rate. Informally, for a complexity class $\calC$ and $X \subseteq \C$, the $\calC$-entropy rate of $X$ is the infimum $s$ for which all elements of $X$ can be covered infinitely often by a language $A \in \calC$ that has $\frac{\log |A_{=n}|}{n} \leq s$ for all sufficiently large $n$. Intuitively, this corresponds to the compression rate when using $A$ as an implicit code, where it takes $\log |A_{=n}|$ bits to specify a member of $A_{=n}$.\\

\noindent
{\bf Kolmogorov Complexity.}
In Section \ref{subsec:kolmogorov_complexity}, we connect
$\SharpP$-measure and $\SharpP$-dimension to Kolmogorov complexity. For a time bound $t(n)$, $K^t(x)$ is the length of the shortest program that prints $x$ in $t(|x|)$ time on a universal Turing machine. In particular, we extend a result from Lutz \cite{Lutz:AEHNC} and show that $\{S \mid (\exists^\infty n) K^p(S \restr n) < n- f(n) \}$  has $\SharpP$-measure 0, where $p$ is a polynomial, and $\sum_{n=0}^\infty 2^{-f(n)}$ is a $\p$-convergent series. In other words, we prove that if $S$ is $\SharpP$-random, then $K^p(S \restr n) \geq n - f(n)$ almost everywhere. We also show that $\SharpP$-dimension is at most the polynomial-time Kolmogorov rate \cite{Hitchcock:phdthesis,Hitchcock:DERC}.
We build on recent work of Nandakumar, Pulari, Akhil S, and Sarma  \cite{NandakumarPulariSSarma24} on Kolmogorov complexity rates and polynomial-time dimension to show that if one-way functions exist, then $\SharpP$-dimension is distinct from $\P$-dimension.

\subsection{Our Results}

In Section \ref{sec:applications} we study the measure and dimension of classical and quantum circuit complexity classes.
Shannon \cite{Shannon49} showed that almost all Boolean functions on $n$ input bits require $(1-\epsilon)\twonn$-size circuits.
Lutz \cite{Lutz:AEHNC} strengthened this in two ways,
showing that the larger circuit complexity class
\[X_\alpha = \SIZEio\!\left(\frac{2^n}{n}\left(1+\frac{\alpha \log n}{n}\right)\right)\]
has $\pspace$-measure 0 for all $\alpha < 1$.
Frandsen and Miltersen \cite{FraMil05} strengthened the original upper bound of Lupanov \cite{Lupa58} to show that Lutz's bound is nearly tight: $X_\alpha$ contains all problems when $\alpha > 3$.

We improve Lutz's result to show that
$\muSpanP(X_\alpha) = 0$ for all $\alpha < 1$.
Lutz's proof extensively reuses polynomial space to consider only {\em novel} circuits, that compute a different function than any previously considered circuit. This proof does not adapt easily to our setting. In Section \ref{sec:entropy_rates_kolmogorov_complexity}, we introduce a measure notion called $\MNP$ that bridges the gap between
$\mupspace$ and $\muSpanP$ by utilizing the Minimum Circuit Size Problem ($\MCSP$) \cite{KabCai00}, allowing the construction of a counting martingale.
As a corollary, we conclude that
$X_\alpha$ has measure 0 in the third level
$\DeltaEthree = \E^\SigmaPtwo$ of the exponential-time hierarchy. While it was previously known how to construct a problem in $\DeltaEthree$ with maximum circuit-size complexity \cite{MiViWa99}, our result says most problems in $\DeltaEthree$ have nearly maximal circuit-size complexity.

Li \cite{Li24}, building on work of Korten \cite{Korten22} and Chen, Hirahara, and Ren \cite{ChenHiraharaRen24}, showed the first exponential-size circuit lower bound within the second level of the exponential-time hierarchy, that the symmetric alternation class $S^\E_2 \not\subseteq \SIZEio(\frac{2^n}{n}).$ We note that Li's proof extends to show
$S^\E_2 \not\subseteq X_\alpha$
for all $\alpha < 1$.
Under suitable derandomization assumptions (see Derandomization Hypothesis \ref{hypothesis:ENPtt_NPSVcircuits}), our $\DeltaEthree$ result improves by one level in the exponential hierarchy, showing $X_\alpha$ has measure 0 in
$\DeltaEtwo = \E^\NP$, for all $\alpha < 1$.

We also improve previous dimension results on circuit-size complexity \cite{Lutz:DCC,Hitchcock:DERC} and the density of hard sets \cite{Lutz:MSDHL,Lutz:DWCPAR,Fu95,Hitchcock:OLRBD}.
Additionally, we use $\GapP$-martingales to apply our counting measure framework to analyze quantum circuit complexity.  We show that the quantum circuit-size class $\BQSIZE\left(o\left(\frac{2^n}{n}\right)\right)$ has $\GapP$-dimension 0.
We also show that the class of
problems that $\Ppoly$-Turing
reduce to subexponentially dense sets has
$\SharpP$-measure 0.
\begin{figure}
	\begin{center}
		\begin{tabular}{|c|c|c|c|c|}
			\hline
			$\SIZEiolutzbound$                                & $\MNP$-measure 0                         & Theorem \ref{th:SIZEio_MNP_measure_0}                  \\
			\hline
			$\SIZEiolutzbound$                                & $\SpanP$-measure 0                       & Corollary \ref{co:SIZEio_SpanP_measure_0}              \\
			\hline
			$\SIZEiolutzbound$                                & $\Deltapthree$-measure 0                 & Corollary \ref{co:lutzbound_DeltaEthree}               \\
			\hline
			$\SIZEdimbound$                                   & $\SharpP$-strong dimension $\alpha$      & Theorem \ref{th:SIZE_dimSharpP}                        \\
			\hline
			$\Ppoly$                                          & $\SharpP$-strong dimension 0             & Corollary \ref{co:Ppoly_SharpP_dimension}              \\
			\hline
			$\SIZEiodimbound$                                 & $\SharpP$-dimension $\frac{1+\alpha}{2}$ & Theorem \ref{th:SIZEio_dimSharpP}                      \\
			\hline
			$\BQSIZE\left(o\left(\frac{2^n}{n}\right)\right)$ & $\GapP$-strong dimension 0               & Theorem \ref{th:BQSIZE_GapP_dimension_0}               \\
			\hline
			$\BQP/\poly$                                      & $\GapP$-strong dimension 0               & Corollary \ref{co:BQPpoly_has_GapP_strong_dimension_0} \\
			\hline
			$(\Ppoly)_\T(\DENSE^c)$                           & $\SharpP$-dimension 0                    & Theorem \ref{th:PPolyT_density_SharpP_dimension_0}     \\
			\hline
		\end{tabular}
	\end{center}
	\caption{Summary of Measure, Dimension, and Strong Dimension Results}\label{figure:summary_of_results}
\end{figure}

See Figure \ref{figure:summary_of_results} for a summary of our results.
We anticipate many further applications of counting measures to refine results where the $\pspace$-measure is known and the $\p$-measure is unknown. We discuss some of these directions and open questions in Section \ref{sec:conclusion}.

\subsection{Organization}

This paper is organized as follows. Section \ref{sec:preliminaries} covers preliminaries. Section \ref{sec:counting_martingales} introduces counting martingales, counting measures, and counting dimensions. In Section \ref{sec:constructions}, we detail the five constructions of counting martingales. Section \ref{sec:entropy_rates_kolmogorov_complexity} presents our tools on entropy rates and Kolmogorov complexity. Our primary applications on circuit complexity are presented in Section \ref{sec:applications}. Finally, Section \ref{sec:conclusion} provides concluding remarks and open questions.

\sectionnewpage
\section{Preliminaries}\label{sec:preliminaries}
The set of all finite binary strings is $\{0,1\}^*$.  The empty string
is denoted by $\lambda$.  We use the standard enumeration of binary
strings $s_0 = \lambda, s_1 = 0, s_2 = 1, s_3 = 00, \ldots$.  For two
strings $x,y \in \{0,1\}^*$, we say $x \leq y$ if $x$ precedes $y$ in
the standard enumeration and $x < y$ if $x$ precedes $y$ and is not
equal to $y$. Given two strings $x$ and $y$, we denote by $[x,y]$ the set of all strings $z$ such that $x\leq z \leq y$. Other types of intervals are defined similarly.
We write $x-1$ for the predecessor
of $x$ in the standard enumeration.  We use the notation $x \prefix y$ to say that
$x$ is a prefix of $y$.  The length of a string $x\in\strings$ is
denoted by $|x|$.

All {\em languages} (decision problems) in this paper are encoded as
subsets of $\{0,1\}^*$.  For a language $A \subseteq \strings$ and $n \geq 0$, we
define $A_{\leq n} = A \cap \{0,1\}^{\leq n}$ and $A_{=n} = A \cap
	\{0,1\}^n$.

The {\em Cantor space} of all infinite binary sequences is $\C$.  For a string $w \in \{0,1\}^*$, the {\em cylinder} $\C_w = w \cdot \C$ consists of all elements of $\C$ that begin with $w$.
We routinely identify a language $A \subseteq \{0,1\}^*$ with the element of Cantor space that is $A$'s characteristic sequence according to the
standard enumeration of binary strings.  In this way, each complexity
class is identified with a subset of Cantor space.  We write $A\restr
	n$ for the $n$-bit prefix of the characteristic sequence of $A$, and
$A[n]$ for the $\nth$-bit of its characteristic sequence.
We use $\log$ for the base 2 logarithm.

Our definitions of most complexity classes are standard \cite{AroraBarak09}. For any function $s : \N \to \N$, $\SIZE(s(n))$ is the class of all
languages $A$ where for all sufficiently large $n$, $A_{=n}$ can be decided
by a circuit with no more than $s(n)$ gates. We write $\SIZEio(s(n))$ for the class of all $A$ where $A_{=n}$ has an $s(n)$-size circuit for infinitely many $n$.

\subsection{Resource-Bounded Measure and Dimension}

Lutz used martingales to define resource-bounded measure
\cite{Lutz:AEHNC} and dimension \cite{Lutz:DCC}.  Athreya et al. \cite{Athreya:ESDAICC} defined strong dimension.
We review the basic definitions.
More background is
available in the survey papers
\cite{Lutz:QSET,Lutz:TPRBM,AmbMay97,Hitchcock:FGCC,Lutz:EFD,Stull20,Mayordomo:EFDAIT,Mayordomo:EHD}.

\begin{definition_named}{Martingale}
	A {\em martingale} is a function $d : \{0,1\}^* \to [0,\infty)$ such
	that for all $w \in \{0,1\}^*$, \[ d(w) = \frac{d(w0) + d(w1)}{2}. \]
	If we relax the equality to a $\geq$ inequality in the above equation, we call
	$d$ a {\em supermartingale}.
\end{definition_named}

\begin{definition_named}{Martingale Success} Let $d$ be a supermartingale or martingale.
	\begin{enumerateC}
		\item We say $d$ {\em succeeds on} a sequence $A \in \C$ if $\limsupn d(A\restr n) = \infty$.
		\item The {\em success set} $S^\infty[d]$ is the class of sequences that $d$ succeeds on.
		\item The {\em unitary success set} of $d$ is the set
		$S^1[d] = \{ A \in \C \mid (\exists n)\ d(A\restr n) \geq 1 \}.$
		\item For $s > 0$, we say  $d$  {\em $s$-succeeds on} $A$ if $(\exists^\infty n)\ d(A \restr n) \geq 2^{(1-s)n}$.
		\item For $s > 0$, we say  $d$  {\em $s$-strongly succeeds on} $A$ if $(\forall^\infty n)\ d(A \restr n) \geq 2^{(1-s)n}$.
		\item We say $d$ {\em $0$-succeeds on} $A$ if $d$ $s$-succeeds on $A$ for all $s > 0$.
		\item We say $d$ {\em $0$-strongly succeeds on} $A$ if $d$ $s$-strongly succeeds on $A$ for all $s > 0$.
		\item For $s \geq 0$, we say $d$ {\em $s$-succeeds on} a class $X \subseteq \C$ if $d$ $s$-succeeds on every member of $X$.
		\item For $s \geq 0$, we say $d$ {\em $s$-strongly succeeds on} a class $X \subseteq \C$ if $d$ $s$-strongly succeeds on every member of $X$.
	\end{enumerateC}
\end{definition_named}

In the following definition $\Delta$ can be any of the time or space resource bounds including $\P$ and
$\pspace$ considered by Lutz \cite{Lutz:AEHNC}, and their relativizations including $\Deltaptwo = \P^\NP$ and $\Deltapthree = \P^\SigmaPtwo$ \cite{Mayo94b,Hitchcock:DERC}.

\begin{definition_named}{Resource-Bounded Measure and Dimension} Let $\Delta$ be a resource bound and let $X
		\subseteq \C$.
	\begin{enumerateC}

		\item  A class $X \subseteq \C$ has {\em $\Delta$-measure 0}, and we write $\mu_\Delta(X) =
			0$, if there is a $\Delta$-computable martingale $d$ with $X
			\subseteq S^\infty[d]$.

		\item  A class $X \subseteq \C$ has {\em $\Delta$-measure 1}, and we write $\mu_\Delta(X) =
			1$, if $\mu_\Delta(X^c) = 0$, where $X^c$ is the complement of $X$ within $\C$.

		\item The {\em $\Delta$-dimension} of a class $X \subseteq \C$ is
		\[\dim_\Delta(X) = \inf \{ s \mid \exists\ \Delta\textrm{-martingale $d$ that $s$-succeeds on all of $X$} \}. \]

		\item The {\em $\Delta$-strong dimension} of a class $X \subseteq \C$ is
		\[ \Dim_\Delta(X) = \inf \{ s \mid \exists\ \Delta\textrm{-martingale $d$ that $s$-strongly succeeds on all of $X$} \}. \]

		\item The {\em $\Delta$-dimension} of a sequence $S \in \C$ is $\dim_\Delta(S) = \dim_\Delta(\{S\})$.
		\item The {\em $\Delta$-strong dimension} of a sequence $S \in \C$ is $\Dim_\Delta(S) = \Dim_\Delta(\{S\})$.

	\end{enumerateC}
\end{definition_named}
\noindent
We note that for all of the classical resource bounds, martingales and supermartingales are equivalent \cite{AmNeTe96}.

\subsection{Counting Complexity}

Valiant introduced $\SharpP$ in the seminal paper for counting complexity \cite{Valiant79:Permanent}.

\begin{definition_named}{$\SharpP$ \cite{Valiant79:Permanent}}
	Let $M$ be a polynomial-time probabilistic Turing machine that accepts or rejects on each computation path. The $\sharpP$ function computed by $M$ is defined as
	\[f(x) = \textrm{number of accepting computation paths of $M$ on input $x$}\]
	for all $x \in \{0,1\}^*$.
\end{definition_named}

K{\"o}bler, Sch{\"o}ning, and Toran \cite{KoScTo89} introduced
$\spanP$ as an extension of $\sharpP$.

\begin{definition_named}{$\SpanP$ \cite{KoScTo89}}  Let $M$ be a polynomial-time probabilistic
	Turing machine that on each computation path either outputs a string
	or outputs nothing.  The $\spanP$ function computed by $M$ is
	defined as \[f(x) = \textrm{number of distinct strings output by $M$ on
			input $x$}\] for all $x \in \{0,1\}^*$.
\end{definition_named}

Every $\sharpP$ function is also a $\spanP$ function.  K{\"o}bler et al. \cite{KoScTo89} showed that $\sharpP = \spanP$ if and only if $\UP=\NP$.
They also extended Stockmeyer's approximate counting \cite{Stoc85} of $\sharpP$ functions in polynomial-time with a $\Sigma^\P_2$ oracle to $\spanP$.
\begin{theorem_cite}{K{\"o}bler, Sch{\"o}ning, and Toran
		\cite{KoScTo89}}\label{th:spanP} Let $f \in \spanP$.  Then there is
	a function $g \in \Deltapthree$ such that for all $n$, for all $x
		\in \{0,1\}^n$,
	$(1-1/n) g(x) \leq f(x) \leq (1+1/n) g(x).$
\end{theorem_cite}

\theoremstyle{definition}
\newtheorem{Hypothesis}[theorem]{Hypothesis}
\newtheorem{DerandHypothesis}[theorem]{Derandomization Hypothesis}

Shaltiel and Umans \cite{ShaUma05} showed that under a derandomization
assumption, $\sharpP$ functions can be approximated by a deterministic
polynomial-time algorithm with nonadaptive access to an $\NP$ oracle. Hitchcock and Vinodchandran \cite{Hitchcock:DERC} noted this extends to $\SpanP$ functions.
\begin{DerandHypothesis}\label{hypothesis:ENPtt_NPSVcircuits}
	$\E^\NP_{\parallel}$ requires exponential-size SV-nondeterministic circuits.
\end{DerandHypothesis}
\noindent We refer to \cite{ShaUma05} for the details of Derandomization Hypothesis \ref{hypothesis:ENPtt_NPSVcircuits}, including equivalent hypotheses. We note that Derandomization Hypothesis \ref{hypothesis:ENPtt_NPSVcircuits} is true under more familiar hypotheses like $\NP$ does not have $\P$-measure 0 \cite{Hitchcock:DERC} or $\E$ requires exponential-size $\NP$-oracle circuits.

\begin{theorem_cite}{\cite{ShaUma05,Hitchcock:DERC}}\label{th:spanP_derand}
	If Derandomization Hypothesis \ref{hypothesis:ENPtt_NPSVcircuits} is true,
	then for any function $f \in \spanP$,
	there is a function $g$ computable in polynomial time with
	nonadaptive access to an $\NP$ oracle such that for all
	$n$, for all $x \in \{0,1\}^n$,
	$g(x) \leq f(x) \leq g(x) (1+1/n).$
\end{theorem_cite}

Fenner, Fortnow, and Kurtz \cite{FeFoKu94} and Li \cite{Li93} introduced the class $\GapP$.

\begin{definition_named}{$\GapP$ \cite{FeFoKu94,Li93}}
	Let $M$ be a polynomial-time probabilistic Turing machine that accepts or rejects on each computation path. The $\GapP$ function computed by $M$ is defined as
	\begin{eqnarray*}
		f(x) &=&
		\textrm{number of accepting computation paths of $M$ on input $x$}\\
		&&  -  \textrm{number of rejecting computation paths of $M$ on input $x$}
	\end{eqnarray*}
	for all $x \in \{0,1\}^*$.
\end{definition_named}
Equivalently, $\GapP$ is the closure of $\SharpP$ under subtraction \cite{FeFoKu94}. Note that every $\SharpP$ function is also a $\GapP$ function and $\P^\SharpP = \P^\GapP$. The class $\SPP$ consists of all languages $L$ where the characteristic function of $L$ is a $\GapP$ function.

\sectionnewpage
\section{Counting Martingales}\label{sec:counting_martingales}

In this section,  we define {\em counting martingales} and use them to define {\em counting measures and dimensions}.  We then work out their foundations including union lemmas, Borel-Cantelli lemmas, and measure conservation that will be used in later sections.

\subsection{Counting Martingales Definitions}

We define {\em counting martingales} as martingales that are the ratio of a counting function from $\SharpP$, $\SpanP$, or $\GapP$ and a polynomial-time function that is always a power of $2$. We consider both approximately computable and exactly computable martingales.

\begin{definition_named}{Counting Martingales} \label{def:counting_martingales}
	Let $\Delta \in \{\SharpP,\SpanP,\GapP\}$ be a counting resource bound.
	\begin{enumerateC}
		\item A {\em $\Delta$-martingale} is a martingale $d(w)$
		where there exist     $f \in \Delta$ and $g \in \FP$ with $g(w, r)$
		being a power of 2 for all $w \in \{0,1\}^*$, such that for all
		$w \in \{0,1\}^*$ and $r \in \N$,
		\[ \left|d(w) - \frac{f(w,r)}{g(w,r)}\right| \leq 2^{-r}. \]
		Here $r$ is encoded in unary.
		\item An {\em exact $\Delta$-martingale} is a martingale
		$$d(w) = \frac{f(w)}{g(w)},$$ where $f \in \Delta$, $g \in \FP$, and $g(w)$
		is a power of 2 for all $w \in \{0,1\}^*$.
	\end{enumerateC}
\end{definition_named}

\subsection{Counting Measures and Dimensions Definitions}

Analogous to the original definitions of resource-bounded measure \cite{Lutz:AEHNC}, we use counting martingales to define counting measures.

\begin{definition_named}{Counting Measure Zero}
	Let $\Delta \in \{\SharpP,\SpanP,\GapP\}$ be a counting resource bound.
	A class $X\subseteq \C$ has {\em $\Delta$-measure 0}, written $\mu_\Delta(X) = 0$, if there is a $\Delta$-martingale $d$ with $X \subseteq S^\infty[d]$.
\end{definition_named}

\begin{definition_named}{Counting Random Sequences}
	Let $\Delta \in \{\SharpP,\SpanP,\GapP\}$ be a counting resource bound.
	A sequence $S \in \C$ is {\em $\Delta$-random} if $\{S\}$ does not have $\Delta$-measure 0.
\end{definition_named}
Equivalently, $S$ is $\Delta$-random if no $\Delta$-martingale succeeds on $S$.
We similarly extend the definitions of resource-bounded dimension \cite{Lutz:DCC} using stricter notions of martingale success.

\begin{definition_named}{Counting Dimensions}
	Let $\Delta \in \{\SharpP,\SpanP,\GapP\}$ be a counting resource bound.
	\begin{enumerate}
		\item The {\em $\Delta$-dimension} of a class $X \subseteq \C$ is
		      $$\dim_\Delta(X) = \inf \{ s \mid \exists\ \Delta\textrm{-martingale $d$ that $s$-succeeds on all of $X$} \}.$$
		\item The {\em $\Delta$-strong dimension} of a class $X \subseteq \C$ is
		      $$\Dim_\Delta(X) = \inf \{ s \mid \exists\ \Delta\textrm{-martingale $d$ that $s$-strongly succeeds on all of $X$} \}. $$
		\item The {\em $\Delta$-dimension} of a sequence $S \in \C$ is $\dim_\Delta(X) = \dim_\Delta(\{S\})$.
		\item The {\em $\Delta$-strong dimension} of a sequence $S \in \C$ is $\Dim_\Delta(X) = \Dim_\Delta(\{S\})$.
	\end{enumerate}
\end{definition_named}
\noindent

The following relationships are immediate.

\begin{proposition} Let $\Delta \in \{\SharpP,\SpanP,\GapP\}$ be a counting resource bound and let $X \subseteq \C$.
	\begin{enumerate}
		\item $0 \leq \dim_\Delta(X) \leq \Dim_\Delta(X) \leq 1$.
		\item If $\dim_\Delta(X) < 1$, then $\mu_\Delta(X) = 0$.
	\end{enumerate}
\end{proposition}

In the following proposition, $\mu$ is Lebesgue measure \cite{Lebesgue1902}, $\dimh$ is Hausdorff dimension \cite{Haus19} and $\dimpack$ is packing dimension \cite{Tric82}.
(We use the notation $\dimpack$ to differentiate it from the polynomial-time dimensions $\dimp$ and $\Dimp$.)
\begin{proposition} \label{prop:basic_counting_measure_dimension_relationships}
	For all $X \subseteq \C$,
	\[\begin{array}{ccccccc}
			\mup(X) = 0 & \Rightarrow & \musharpp(X)=0  & \Rightarrow & \mugapp(X) = 0                               \\
			            &             & \Downarrow      &             & \Downarrow                                   \\
			            &             & \muspanp(X) = 0 & \Rightarrow & \mupspace(X) = 0 & \Rightarrow & \mu(X) = 0,
		\end{array}\]
	\[
		\begin{array}{ccccccc}
			0 \leq \dimh(X) & \leq & \dimpspace(X) & \leq & \dimGapP(X)                             \\
			                &      & \rotleq       &      & \rotleq                                 \\
			                &      & \dimSpanP(X)  & \leq & \dimSharpP(X) & \leq & \dimp(X) \leq 1,
		\end{array}
	\]
	and
	\[
		\begin{array}{ccccccc}
			0 \leq \dimpack(X) & \leq & \Dimpspace(X) & \leq & \DimGapP(X)                             \\
			                   &      & \rotleq       &      & \rotleq                                 \\
			                   &      & \DimSpanP(X)  & \leq & \DimSharpP(X) & \leq & \Dimp(X) \leq 1.
		\end{array}
	\]
\end{proposition}
\begin{proof}
	By the Exact Computation Lemma \cite{Lutz:WCE}, the
	outputs of a $\P$-martingale $d$ may be expressed as dyadic rationals of the form $\frac{n}{2^m}$.  It is easy to see that
	the numerator is computed by a $\SharpP$ function. The polynomial-time PTM $M(s,w)$ associated with the $\SharpP$ function takes input a string $s$ and witness $w$ and computes $d(s) = \frac{n}{2^m}$ in $\poly(s)$ time. If $w$ encodes an integer in $[1,n]$ with no leading zeros, then $M(s,w)$ accepts. Then $M(s,w)$ has $n$ accepting paths, and the denominator $2^m$ is polynomial-time computable. This implies that every exactly computable $\P$-martingale is also a $\SharpP$-martingale. The other relationships follow by complexity class containments and the martingale characterizations of Lebesgue measure \cite{Vill39}, Hausdorff dimension \cite{Lutz:DCC}, and packing dimension \cite{Athreya:ESDAICC}.
\end{proof}

\subsection{Basic Properties of Counting Martingales}

We first note that counting martingales are closed under finite sums, which implies finite unions of counting measure 0 sets have counting measure 0.
\begin{lemma}\label{le:finite_union}
	Let $\Delta \in \{\SharpP,\SpanP,\GapP\}$ be a counting resource bound.
	Exact $\Delta$-martingales are closed under finite sums.
\end{lemma}
\begin{appendixproof}[Proof of Lemma \ref{le:finite_union}]
	Let $d_1 = \frac{f_1}{g_1}$ and $d_2 = \frac{f_2}{g_2}$ be exact
	$\Delta$-martingales. We have
	\[ d_1(w)+d_2(w) = \frac{f_1(w)g_2(w)+f_2(w)g_1(w)}{g_1(w)g_2(w)}. \]
	The numerator is a $\Delta$-function by closure properties of $\Delta$ and the denominator is an $\FP$
	function that is always a power of 2.

	We can be more efficient because the denominators are powers of $2$. Suppose $g_2(w) \leq g_1(w)$. Then $\frac{g_1(w)}{g_2(w)}$ is a power
	of 2.
	Therefore \[d_1(w)+d_2(w) = \frac{f_1(w)+f_2(w)\frac{g_1(w)}{g_2(w)}}{g_1(w)}.\]
	The case $g_1(w) < g_2(w)$ is analogous.
\end{appendixproof}

\begin{corollary}\label{co:finite-unions}
	Let $\Delta \in \{\SharpP,\SpanP,\GapP\}$ be a counting resource bound and let
	$X, Y \subseteq \C$.
	\begin{enumerateC}
		\item If $\mu_{\Delta}(X) = 0$ and $\mu_{\Delta}(Y) = 0$, then $\mu_{\Delta}(X \cup Y) = 0.$
		\item $\dim_\Delta(X \cup Y) = \max\{ \dim_\Delta(X), \dim_\Delta(Y) \}.$
		\item $\Dim_\Delta(X \cup Y) = \max\{ \Dim_\Delta(X), \Dim_\Delta(Y) \}.$
	\end{enumerateC}

\end{corollary}

Lutz showed that uniform countable unions of $\Delta$-measure 0 sets have $\Delta$-measure 0 for time and space resource bounds $\Delta$. This was proved by summing martingales. We establish an analogue for a uniform family of exact counting martingales.

\begin{definition_named}{Uniform Family of Exact Counting Martingales}
	Let $\Delta \in \{\SharpP,\SpanP,\allowbreak\GapP\}$ be a counting resource bound.
	We say that a family $\left(d_n=\frac{f_n}{g_n}\mid n \in \N\right)$ of exact $\Delta$-martingales is uniform if $(f_n \mid n \in \N)$ is uniformly $\Delta$-computable and $(g_n\mid n \in \N)$ is uniformly $\FP$.
\end{definition_named}

\begin{definition_named}{$\P$-convergence}
	A series $\sum_{n=0}^{\infty} a_n$ of nonnegative real numbers $a_n$ is $\P$-convergent if there is a polynomial-time function $m:\mathbb{N} \to \mathbb{N}$  with
	$\sum_{n=m(i)}^{\infty} a_n \leq 2^{-i} \quad \text{for all } i \in \mathbb{N}.$
	Such a function $m$ is  called a modulus of the convergence.
	A sequence
	$\sum_{k=0}^{\infty} a_{j,k} \quad (j = 0, 1, 2, \ldots)$
	of series of nonnegative real numbers is uniformly \( \Delta \)-convergent if there is a function \( m:\mathbb{N}^2 \to \mathbb{N} \) such that \( m \in \Delta \) and, for all \( j \in \mathbb{N} \), \( m_j \) is a modulus of the convergence of the series \( \sum_{k=0}^{\infty} a_{j,k} \).

\end{definition_named}

\begin{lemma_named}{Counting Martingale Summation Lemma}
	\label{lemma:counting-martingale-summation}
	Let $\Delta \in \{\SharpP,\SpanP,\GapP\}$ be a counting resource bound.
	Suppose $\left(d_n=\frac{f_n}{g_n}\mid n \in \N\right)$
	is a uniform family of
	exact $\Delta$-martingales with $\sum\limits_{n=0}^\infty d_n(w)$ uniformly
	$\p$-convergent for all $w \in \{0,1\}^*$.
	Then $d(w) = \sum\limits_{n=0}^\infty d_n(w)$ is a $\Delta$-martingale.
\end{lemma_named}
\begin{appendixproof}[Proof of Lemma \ref{lemma:counting-martingale-summation}]
	Let $m(w,r)$ be the modulus of $\p$-convergence for $\sum\limits_{n=0}^\infty d_n(w)$.
	Let $w \in \{0,1\}^*$ and $r \in \N$.
	Define
	\[ t(w,r) = \CCfont{max}(g_1(w),\ldots,g_{m(w,r)}(w)) = \lcm(g_1(w),\ldots,g_{m(w,r}(w)). \]
	Define
	\[ \hat{d}(w,r) = \sum_{n=0}^{m(w,r)} d_n(w). \]
	Then \[ |d(w) - \hat{d}(w,r)| = \sum_{n=m(w,r)+1}^\infty d_n(w) \leq 2^{-r}. \]
	We have
	\[ \hat{d}(w,r) = \frac{\sum\limits_{n=0}^{m(w,r)} f_n(w) \cdot \frac{t(w,r)}{g_n(w)}}{t(w,r)}. \]
	This is a ratio of a $\Delta$ function and an $\FP$ function that is a
	power of 2.
\end{appendixproof}

We now have our countable union lemmas.

\begin{lemma_named}{Counting Measure Union Lemma}
	\label{lemma:counting-measure-union}
	Let $\Delta \in \{\SharpP,\SpanP,\GapP\}$ be a counting resource bound.
	Suppose $\left(d_n=\frac{f_n}{g_n} \mid n \in \N\right)$ is a uniform family of
	exact $\Delta$-martingales with $\sum\limits_{n=0}^\infty d_n(w)$  uniformly
	$\p$-convergent for all $w \in \{0,1\}^*$.
	Then $\bigcup\limits_{n=0}^\infty S^\infty[d_n]$ has $\Delta$-measure 0.
\end{lemma_named}
\begin{proof}
	This is immediate from Lemma \ref{lemma:counting-martingale-summation}.
\end{proof}

\begin{lemma_named}{Counting Dimension Union Lemma}
	\label{lemma:counting-dimension-union}
	Let $\Delta \in \{\SharpP,\SpanP,\GapP\}$ be a counting resource bound and let $s > 0$.
	Suppose $\left(d_n=\frac{f_n}{g_n} \mid n \in \N\right)$ is a uniform family of
	exact $\Delta$-martingales with $\sum\limits_{n=0}^\infty d_n(w)$ uniformly
	$\p$-convergent for all $w \in \{0,1\}^*$.
	\begin{enumerate}
		\item Suppose $X_0, X_1, \ldots$ are classes where each $d_n$ $s$-succeeds on $X_n$.
		      Then $\bigcup\limits_{n=0}^\infty X_n$ has $\Delta$-dimension at most $s$.
		\item Suppose $X_0, X_1, \ldots$ are classes where each $d_n$ $s$-strongly succeeds on $X_n$.
		      Then $\bigcup\limits_{n=0}^\infty X_n$ has $\Delta$-strong dimension at most $s$.
	\end{enumerate}
\end{lemma_named}
\begin{proof}
	This is immediate from Lemma \ref{lemma:counting-martingale-summation}.
\end{proof}

\subsection{Borel-Cantelli Lemmas}

Lutz \cite{Lutz:AEHNC} proved a resource-bounded Borel-Cantelli Lemma. We now
present the counting measure version.

\begin{lemma_named}{Counting Measure Borel-Cantelli Lemma}\label{lemma:counting-measure-borel-cantelli}
	Let $\Delta \in \{\SharpP,\SpanP,\GapP\}$ be a counting resource bound.
	Suppose $(d_n=\frac{f_n}{g_n} \mid n \in \N)$ is a uniform family of exact
	$\Delta$-martingales with $\sum_{n=0}^\infty d_n(w)$ being uniformly
	$\Delta$-convergent for all $w \in \{0,1\}^*$.
	Then
	\[ \limsupn S^1[d_n] =
		\bigcap_{i=0}^\infty \bigcup_{j \geq i}^\infty S^1[d_j] = \{ S \in \C \mid (\exists^\infty n)\ S \in S^1[d_n] \}
	\] has $\Delta$-measure 0.
\end{lemma_named}
\begin{appendixproof}[Proof of Counting Measure Borel-Cantelli Lemma (Lemma \ref{lemma:counting-measure-borel-cantelli})]
	By Lemma \ref{lemma:counting-martingale-summation}, the exact $\Delta$-martingale family sums to $\Delta$-martingale $d$. Then $\bigcap\limits_{i=0}^\infty \bigcup\limits_{j \geq i}^\infty S^1[d_j] \subseteq S^\infty[d]$.
\end{appendixproof}

Analogously, we extend  Lutz's Borel-Cantelli lemma for time- and space-bounded dimension \cite{Lutz:DCC} to counting dimensions.

\begin{lemma_named}{Counting Dimension Borel-Cantelli Lemma}
	\label{lemma:counting-dimension-borel-cantelli}
	Let $\Delta \in \{\SharpP,\SpanP,\GapP\}$ be a counting resource bound and let $s > 0$.
	Suppose $\left(d_n=\frac{f_n}{g_n} \mid n \in \N\right)$ is a uniform family of
	exact $\Delta$-martingales with $d_n(\lambda) \leq 2^{(s-1)n}$.
	Then
	\[ \limsupn S^1[d_n] =
		\bigcap_{i=0}^\infty \bigcup_{j \geq i}^\infty S^1[d_j] = \{ S \in \C \mid (\exists^\infty n)\ S \in S^1[d_n] \}
	\] has $\Delta$-dimension at most $s$
	and
	\[ \liminfn S^1[d_n] =
		\bigcup_{i=0}^\infty \bigcap_{j \geq i}^\infty S^1[d_j] = \{ S \in \C \mid (\forall^\infty n)\ S \in S^1[d_n] \}
	\] has $\Delta$-strong dimension at most $s$.
\end{lemma_named}
\begin{proof}
	Let $t > s$ and $0 < \epsilon < t-s$ be rational numbers. Define $$d_n'(w) = 2^{\ceil{(1-t)n}}d_n(w)$$ for all $n \in \N$ and $w \in \{0,1\}^*$. Then $d_n'$ is a uniform family of exact $\Delta$-martingales with $$d_n'(\lambda) \leq
		2^{(s-1)n} 2^{(1-t)n+1} =
		2^{-(t-s)n+1} < 2^{-\epsilon n},$$
	with the last inequality holding for sufficiently large $n$.
	Therefore $d' = \sum\limits_{n=0}^\infty d_n'$ is a $\Delta$-martingale by Lemma \ref{lemma:counting-martingale-summation}. When $d_n(w) \geq 1$, we have $d'(w) \geq d_n'(w) \geq 2^{(1-t)n}$. Therefore $d'$ $t$-succeeds on $\limsupn S^1[d_n]$ and $d'$ $t$-strongly succeeds on $\liminfn S^1[d_n]$.
\end{proof}

\subsection{Measure Conservation}\label{sec:measure_conservation}

Lutz \cite{Lutz:AEHNC} defined a {\em constructor} to be a function $\delta : \{0,1\}^* \to \{0,1\}^*$ that properly extends its input string. The {\em result} $R(\delta)$ of constructor $\delta$ is the infinite sequence obtained by repeatedly applying $\delta$ to the empty string. For a resource bound $\Delta$, the {\em result class} $R(\Delta)$ is $R(\Delta) = \{ R(\delta) \mid \delta \in \Delta \}$.
Lutz's Measure Conservation Theorem \cite{Lutz:AEHNC} showed that $R(\Delta)$ does not have $\Delta$-measure 0 for the time and space resource bounds. Our counting measures do not fit into this framework. It is not clear what a $\SharpP$ constructor would be. The best measure conservation theorem we can prove using a constructor approach for counting measures is the following.
\begin{theorem}\label{th:counting_measure_conservation}
	\begin{enumerateC}
		\item $\E^{\SharpP}$ does not have $\SharpP$-measure 0.
		\item $\E^{\SpanP}$ does not have $\SpanP$-measure 0.
		\item $\E^{\SharpP}= \E^\GapP$ does not have $\GapP$-measure 0.
	\end{enumerateC}
\end{theorem}
\begin{appendixproof}[Proof of Theorem \ref{th:counting_measure_conservation}]
	Let $d$ be a $\SharpP$ martingale. We recursively construct a language
	$L\in\E^{\SharpP}$ that $d$ does not succeed on. Let $x$ be any length $n$
	string and $w$ be the characteristic string for all the strings that
	lexicographically come before $x$. Now we specify the characteristic bit of
	$x$. The string $x$ belongs to $L$ if and only if $d(w1) < d(w0)$. Since
	$d(wb)$ can be computed by a call to a $\SharpP$  oracle and computing a
	polynomial-time function on a length $\Theta(2^n)$ string, we can decide any $x$ in $\E^\SharpP$. Since $d$ cannot grow on
	$L\in\E^\SharpP$, it follows that $\E^\SharpP$ does not have $\SharpP$-measure 0. If $d$ is a $\SpanP$ martingale, we can construct a language $L\in\E^{\SpanP}$. If $d$ is a $\GapP$ martingale, we can construct a language $L\in\E^{\GapP} = \E^{\SharpP}$.
\end{appendixproof}

In the proof of Theorem \ref{th:counting_measure_conservation} the constructor we obtain from a $\SharpP$-martingale is computable in $\P^\SharpP$, resulting in the $\E^\SharpP$ upper bound. We will use approximate counting to improve this to the class $\DeltaEthree = \E^\SigmaPtwo$. By padding Toda's theorem \cite{Toda91,Book74b}, $\DeltaEthree \subseteq \E^\SharpP$. Under suitable derandomization assumptions, we get an improvement to $\DeltaEtwo = \E^\NP$.
The results in the remainder of this section hold not only for $\SharpP$ but for the larger class $\SpanP$. It is open whether $\GapP$ can be approximately counted in the same way, so Theorem \ref{th:counting_measure_conservation} is the best we have for $\GapP$.

\begin{lemma}\label{le:spanp_to_deltapthree_martingale}
	\begin{enumerateC}
		\item For every $\SpanP$-martingale $d$, there is a $\Deltapthree$-supermartingale $d'$ and a $\gamma > 0$ such that $d'(w) \geq \gamma d(w)$ for all $w \in \{0,1\}^*$.
		\item If Derandomization Hypothesis \ref{hypothesis:ENPtt_NPSVcircuits} is true, then for every $\SpanP$-martingale $d$, there is a $\Deltaptwo$-supermartingale $d'$ and a $\gamma > 0$ such that $d'(w) \geq \gamma d(w)$ for all $w \in \{0,1\}^*$.
	\end{enumerateC}
\end{lemma}
\begin{appendixproof}[Proof of Lemma \ref{le:spanp_to_deltapthree_martingale}]
	Let $d = \frac{f}{g}$ be a $\SpanP$-martingale.
	\newcommand{\epsn}{\epsilon}
	Let $h \in \Deltapthree$ be the approximation of $f$ from Theorem
	\ref{th:spanP}.

	\renewcommand{\epsn}{\epsilon_n}
	For each $n$, let $\epsn = \frac{1}{n}$ and define a function $d_n$ by
	\[ d_n(v) =
		\frac{h(v)}{g(v)} \left(\frac{1-\epsn}{1+\epsn}\right)^{n} \]
	for all $v \in \{0,1\}^{\geq n}$ and
	$d_n(v) = d_n(v \restr n)$
	for all $v$ with $|v| > n$. Then $d_n$ is $\Deltapthree$ exactly computable.
	For $v \in \{0,1\}^{<n}$, we have
	\begin{eqnarray*}
		\left(\frac{h(v0)}{g(v0)}+\frac{h(v1)}{g(v1)}\right)
		&\leq& \left(\frac{f(v0)}{g(v0)}+\frac{f(v1)}{g(v1)}\right)\frac{1}{1-\epsn} \\
		&=& \left(d(v0)+d(v1)\right)\frac{1}{1-\epsn} \\
		&=&2d(v)\frac{1}{1-\epsn}\\
		&=&\frac{2f(v)}{g(v)}\frac{1}{1-\epsn}\\
		&\leq& \frac{2h(v)}{g(v)}\frac{1+\epsn}{1-\epsn}.
	\end{eqnarray*}
	Therefore
	\begin{eqnarray*}
		d_n(v0) + d_n(v1)
		&=&
		\left(\frac{h(v0)}{g(v0)}+\frac{h(v1)}{g(v1)}\right)\left(\frac{1-\epsn}{1+\epsn}\right)^{n+1} \\
		\\&\leq&\frac{2h(v)}{g(v)}\frac{1+\epsn}{1-\epsn}
		\left(\frac{1-\epsn}{1+\epsn}\right)^{n+1} \\
		&=&
		\frac{2h(v)}{g(v)}
		\left(\frac{1-\epsn}{1+\epsn}\right)^{n}\\
		&=& 2 d_n(v),
	\end{eqnarray*}
	for all $v \in \{0,1\}^{< n}$, so $d_n$ is a supermartingale.

	Let $v \in \{0,1\}^n$. Since
	\[  \left(\frac{1-\epsn}{1+\epsn}\right)^{n} = \left(1-\frac{2}{n+1}\right)^n \to \frac{1}{e^2}
	\]as $n \to \infty$, let $\gamma \in (0,\tfrac{1}{e^2})$.
	We have
	\[ d_n(v) = d(v) \left(\frac{1-\epsn}{1+\epsn}\right)^{n} \geq \gamma d(v).
	\]
	when $n$ is sufficiently large.

	For part 2, under Derandomization Hypothesis \ref{hypothesis:ENPtt_NPSVcircuits}, we obtain the approximation  $h \in \Deltaptwo$ from Theorem \ref{th:spanP_derand} and follow the same proof.
\end{appendixproof}

The following two theorems and their corollaries are immediate from the previous lemma.

\begin{theorem}\label{th:counting_measure_conservation_deltaEthree}
	Let $X \subseteq \C$.
	\begin{enumerateC}
		\item If $\muSpanP(X) = 0$, then $\muDeltaPthree(X) = 0$.
		\item $\dim_\DeltaPthree(X) \leq \dim_\SpanP(X).$
		\item $\Dim_\DeltaPthree(X) \leq \Dim_\SpanP(X).$
	\end{enumerateC}
\end{theorem}

\begin{corollary}\label{co:DeltaEthree_does_not_have_SpanP_measure_0}
	$\DeltaEthree$ does not have $\SpanP$-measure 0.
\end{corollary}

\begin{theorem}\label{th:counting_measure_conservation_deltaEthree_derand}
	Assume Derandomization Hypothesis \ref{hypothesis:ENPtt_NPSVcircuits}.
	Let $X \subseteq \C$.
	\begin{enumerateC}
		\item If $\muSpanP(X) = 0$, then $\mu_\DeltaPtwo(X) = 0$.
		\item   $\dim_\DeltaPtwo(X) \leq \dim_\SpanP(X).$
		\item   $\Dim_\DeltaPtwo(X) \leq \Dim_\SpanP(X).$
	\end{enumerateC}
\end{theorem}

\begin{corollary}
	If Derandomization Hypothesis \ref{hypothesis:ENPtt_NPSVcircuits} is true, then
	$\DeltaEtwo = \E^\NP$ does not have $\SpanP$-measure 0.
\end{corollary}

We conclude this section by noting that $\P^{\SharpP}$-measure and the $\P^{\SharpP}$-dimensions dominate the counting measures and dimensions. This is immediate from Toda's theorem \cite{Toda91} that $\PH \subseteq \P^{\SharpP} = \P^{\GapP}$
and Theorem \ref{th:counting_measure_conservation_deltaEthree}. It appears, however, that the $\P^\SharpP$ notions are much stronger than our counting dimensions and measures.

\begin{corollary}\label{co:PtoSharpP_dominates_counting_measures_and_dimensions}
	Let $X \subseteq \C$.
	\begin{enumerateC}
		\item If $\muSpanP(X) = 0$, then $\muPSharpP(X) = 0$.
		\item If $\muGapP(X) = 0$, then $\muPSharpP(X) = 0$.
		\item $\dim_\PSharpP(X) \leq \dim_\SpanP(X)$.
		\item $\Dim_\PSharpP(X) \leq \Dim_\SpanP(X)$.
	\end{enumerateC}
\end{corollary}

\section{Counting Martingale Constructions}\label{sec:constructions}

In this section we present five techniques for constructing counting martingales: Cover Martingale, Conditional Expectation Martingale, Subset Martingale, Acceptance Probability Martingale, and Bi-immunity Martingale.

\subsection{Cover Martingale Construction}

The first construction uses a cover set $A$ and defines the martingale using the conditional probability that an extension of the current string is in the cover. The goal is to obtain a value of 1 on nodes in $A$ while having a small initial value at the root $\lambda$.

\newcommand{\ext}{\mathrm{ext}}

\begin{construction_named}{Cover Martingale}\label{construction:cover_martingale}
	Let
	$A \subseteq \{0,1\}^*$ and $n \geq 0$.
	Choose $x$ uniformly at random from $\{0,1\}^n$ and let
	\[ d_n(w) =  \Pr[ x \in A_{=n} \mid w \prefix x] \]
	for all $w \in \{0,1\}^{\leq n}$. For all $w \in \{0,1\}^{> n}$, we let $d_n(w) = d_n(w\restr n)$ take the value of its length-$n$ prefix.
	Then we have \[ d_n(\lambda) = \Pr[x \in A_{=n}], \] and \[ d_n(x) = 1 \] for all $x \in A_{=n}$. For all $x \in \{0,1\}^n - A_{=n}$, note that $d_n(x) = 0$.
	Therefore $$S^1[d_n] = \bigcup_{w \in A_{=n}} \C_w.$$ In other words, $d_n$ covers all sequences that have a prefix in $A_{=n}$ with a value of 1.
	See Figure \ref{figure:martingale} for an example with $n=4$, where the green nodes are $A_{=n}$.
\end{construction_named}

Depending on the complexity of the cover, we have a uniform family of exact counting martingales.

\begin{lemma}\label{le:cover_martingale_construction}
	\begin{enumerateC}
		\item If $A \in \UP$, then Construction \ref{construction:cover_martingale} produces a uniform family of exact $\SharpP$-martingales.
		\item If $A \in \NP$, then Construction \ref{construction:cover_martingale} produces a uniform family of exact $\SpanP$-martingales.
		\item If $A \in \SPP$, then Construction \ref{construction:cover_martingale} produces a uniform family of exact $\GapP$-martingales.
	\end{enumerateC}
\end{lemma}
\begin{appendixproof}[Proof of Lemma \ref{le:cover_martingale_construction}]
	For $w \in \{0,1\}^*$ and $n \geq 0$, define $\ext(w,n) = \{ x \in \{0,1\}^n \mid w \prefix x \}$ and
	$\ext_A(w,n) = A \cap \ext(w,n)$. Then
	\begin{equation}\label{eq:cover_martingale}
		d_n(w) = \frac{|\ext_A(w,n)|}{2^{n - |w|}}
	\end{equation}
	for all
	$w \in \binary^{\leq n}.$
	We have
	\begin{eqnarray*}
		d_n(w0)+d_n(w1) &=&
		\frac{|\ext_A(w0,n)|}{2^{n - |w0|}}
		+
		\frac{|\ext_A(w1,n)|}{2^{n - |w1|}}\\
		&=&
		\frac{|\ext_A(w,n)|}{2^{n - (|w|+1)}}\\
		&=& 2d(w),
	\end{eqnarray*}
	for all $w \in \{0,1\}^{<n}$ and $d_n(w0)+d_n(w1)=2d_n(w \restr n)=2 d_n(w)$ for all $w \in \{0,1\}^{\geq n}$, so $d_n$ is a martingale.

	If $A \in \UP$, then the numerator $|\ext_A(w,n)|$ in \eqref{eq:cover_martingale} is computed by the $\SharpP$ function $M(0^n,w)$ that guesses an extension $x \in \ext(w,n)$, guess a witness $v$ for $x$,  and accepts if $v$ is a valid witness for $x \in A$. Because $A$ has unique witnesses, there is exactly one accepting computation path for each $x \in \ext(w,n)$.

	If $A \in \NP$, then we compute the numerator $|\ext_A(w,n)|$ in \eqref{eq:cover_martingale} by the $\SpanP$ function $M(0^n,w)$ that guesses an extension $x \in \ext(w,n)$, guesses a witness $v$ for $x$, and prints $\pair{x,v}$ if $v$ is a valid witness for $x \in A$.

	If $A \in \SPP$, let $M$ be a $\PTM$ such that $M(x)$ has gap 1 when $x \in A$ and $M(x)$ has gap 0 when $x \not\in A$.
	Consider the $\GapP$ function $N(0^n,w)$ that guesses an extension $x \in \ext(w,n)$ and runs $M(x)$. If $M(x)$ accepts, then $N$ accepts. If $M(x)$ rejects, then $N$ rejects. The gap of $N(0^n,w)$ is $|\ext_A(w,n)|$.
\end{appendixproof}

We will use Construction \ref{construction:cover_martingale} in Section \ref{subsec:entropy_rates} to develop tools relating counting dimensions and entropy rates, with applications in Section \ref{sec:applications} to circuit complexity.

\subsection{Conditional Expectation Martingale Construction}

Here is a more general martingale construction using a counting function $f(x)$ as a random variable and taking the conditional expectation of $f$ given the current prefix $w$. This generalizes Construction \ref{construction:cover_martingale} when $f(x)$ is the indicator random variable for the membership of $x$ in the cover $A$.

\begin{construction_named}{Conditional Expectation Martingale}\label{construction:random_variable_martingale}
	Let $f : \{0,1\}^* \to \N$
	and view $f(x)$ as a random variable where $x$ is chosen uniformly from $\{0,1\}^n$. Define
	\[ d_n(w) = E[f(x) \mid w \prefix x] \]
	for all $w \in \{0,1\}^{\leq n}$. For all $w \in \{0,1\}^{> n}$, we let $d_n(w) = d_n(w\restr n)$ take the value of its length-$n$ prefix.
	Then
	\[ d_n(\lambda) = E[f(x)], \] and \[ d_n(x) = f(x) \] for all $x \in \{0,1\}^n$.
	See Figure \ref{figure:martingale_random_variable}  for an example with $n=4$. The green nodes are the $x \in \{0,1\}^n$ that have $f(x) > 0$.
\end{construction_named}

We think of Construction \ref{construction:random_variable_martingale} as covering the sequences that have a prefix $x$ with $f(x) > 0$. Depending on the counting complexity of $f$, we have a uniform family of exact counting martingales.

\begin{lemma} \label{le:random_variable_martingale_construction}
	\begin{enumerateC}
		\item If $f \in \SharpP$, then Construction \ref{construction:random_variable_martingale}
		produces a uniform family of exact $\SharpP$-martingales.
		\item If $f \in \SpanP$, then Construction \ref{construction:random_variable_martingale}
		produces a uniform family of exact $\SpanP$-martingales.
		\item If $f \in \GapP$, then Construction \ref{construction:random_variable_martingale}
		produces a uniform family of exact $\GapP$-martingales.
	\end{enumerateC}
\end{lemma}
\begin{appendixproof}[Proof of Lemma \ref{le:random_variable_martingale_construction}]
	Let $\ext(w,n) = \{ x \in \{0,1\}^n \mid w \prefix x \}$.
	Then
	\begin{equation}
		d_n(w) = \frac{\sum\limits_{x \in \ext(w,n)} f(x)}{2^{n - |w|}}  \label{eq:random_variable_martingale}
	\end{equation}
	for all
	$w \in \binary^{\leq n}.$
	We have
	\begin{eqnarray*}
		d_n(w0)+d_n(w1) &=&
		\frac{\sum\limits_{x \in \ext(w0,n)} f(x)}{2^{n - |w0|}} +
		\frac{\sum\limits_{x \in \ext(w1,n)} f(x)}{2^{n - |w1|}}\\
		&=& \frac{\sum\limits_{x \in \ext(w,n)} f(x)}{2^{n - (|w|+1)}} \\
		&=& 2d(w),
	\end{eqnarray*}
	for all $w \in \{0,1\}^{<n}$ and $d_n(w0)+d_n(w1)=2d_n(w \restr n)=2 d_n(w)$ for all $w \in \{0,1\}^{\geq n}$, so $d_n$ is a martingale.

	Suppose $f \in \SharpP$. We compute the numerator of $d_n(w)$ in \eqref{eq:random_variable_martingale} by the PTM $M(0^n,w)$ that guesses an extension $x \in \ext(w,n)$ and then runs the $\SharpP$ algorithm for $f$ on $x$. If $f(x)$ accepts, then $M$ accepts. If $f(x)$ rejects, then $M$ rejects.

	Suppose $f \in \SpanP$. We compute the numerator of $d_n(w)$ in \eqref{eq:random_variable_martingale} by the PTM $N(0^n,w)$ that guesses an extension $x \in \ext(w,n)$ and then runs the $\SpanP$ algorithm for $f$ on $x$. If $f(x)$ has an output $v$, then $N$ prints $\pair{x,v}$.

	Suppose $f \in \GapP$. We compute the numerator of $d_n(w)$ in \eqref{eq:random_variable_martingale} by the PTM $G(0^n,w)$ that guesses an extension $x \in \ext(w,n)$ and then runs the $\GapP$ algorithm for $f$ on $x$. If $f(x)$ accepts, then $G$ accepts. If $f(x)$ rejects, then $G$ rejects.
\end{appendixproof}

We will use Construction \ref{construction:random_variable_martingale} in Section \ref{subsec:kolmogorov_complexity} to develop tools relating counting measures and dimensions to Kolmogorov Complexity, with applications to circuit complexity in Section \ref{sec:applications}.

\subsection{Subset Martingale Construction}

Next, we present a construction that builds upon Construction \ref{construction:cover_martingale}.
For a language $B \subseteq \{0,1\}^*$, the {\em census function of $B$} is defined by $c_B(n) = | B \cap [s_0,s_n) |$ for all $n \geq 0$. For a string $w \in \{0,1\}^*$, let $$L(w) = \{ s_i \mid w[i] = 1 \}$$ be the language with characteristic string $w$.

\begin{construction_named}{Subset Martingale}\label{construction:subset_martingale}
	Let $B \subseteq \{0,1\}^*$ and define the cover
	\begin{eqnarray*}
		A &=& \{ w \in \{0,1\}^* \mid L(w) \subseteq B \} \\
		&=& \{ w \in \{0,1\}^* \mid (\forall i < |w|)\ w[i]=1 \Rightarrow s_i \in B \}.
	\end{eqnarray*}
	Then apply Construction \ref{construction:cover_martingale}. We have $|A_{=n}| = 2^{c_B(n)}$, so \[ d_n(\lambda) = \Pr[ x \in A_{=n} ] = 2^{c_B(n)-n}. \]
	For all $w \in \{0,1\}^n$,
	\[ d_n(w) = \begin{cases} 1 & \textrm{if }L(w) \subseteq B \\
              0 & \textrm{otherwise.}
		\end{cases} \]
	\noindent

	\begin{figure}
		\begin{center}
			\begin{tikzpicture}[level distance=1.5cm,
					level 1/.style={sibling distance=12cm},
					level 2/.style={sibling distance=6cm},
					level 3/.style={sibling distance=3cm},
					level 4/.style={sibling distance=1.5cm},scale=.47,
					every node/.style={draw, rectangle, align=center}]
				\node[fill=green]{$\frac{1}{4}$}
				child {node [fill=green] {$\frac{1}{2}$}
						child {node [fill=green] {$\tfrac{1}{2}$}
								child {node [fill=green] {1}
										child {node [fill=green] {1}} child {node [fill=green] {1}} }
								child {node {0}
										child {node {0}} child {node {0}} }
							}
						child {node [fill=green] {$\tfrac{1}{2}$}
								child {node [fill=green] {1}
										child {node [fill=green] {1}} child {node [fill=green] {1}} }
								child {node {0}
										child {node {0}}
										child {node {0}}
									}
							}
					}
				child {node {0}
						child {node {0}
								child {node {0}
										child {node {0}}
										child {node {0}}
									}
								child {node {0}
										child {node {0}}
										child {node {0}}
									}
							}
						child {node {0}
								child {node {0}
										child {node {0}}
										child {node {0}}
									}
								child {node {0}
										child {node {0}}
										child {node {0}}
									}
							}
					};
			\end{tikzpicture}
		\end{center}\caption{Subset Martingale Construction (Construction \ref{construction:subset_martingale})}\label{figure:subset_martingale}
	\end{figure}
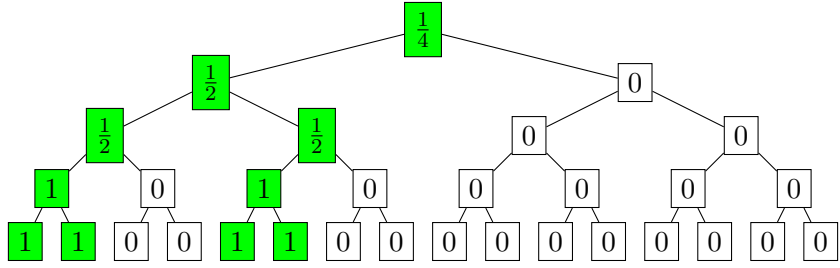 See Figure \ref{figure:subset_martingale} for an example with $n = 4$ and $B \cap [s_0,s_3] = \{s_1,s_3\}$. The nodes $w \in \{0,1\}^{\leq 4}$ that are colored green in the tree are those with $L(w) \subseteq B$.
\end{construction_named}

In the following lemma, $\UE$ is the exponential ($2^{O(n)}$ time) version of $\UP$, $\NE$ is the exponential version of $\NP$, and $\SPE$ is the exponential version of $\SPP$.

\begin{lemma}\label{le:subset_martingale_uniform_family}
	\begin{enumerateC}
		\item If $B \in \UE$, then Construction \ref{construction:subset_martingale} produces a uniform family of $\SharpP$-martingales.
		\item If $B \in \NE$, then Construction \ref{construction:subset_martingale} produces a uniform family of $\SpanP$-martingales.
		\item If $B \in \SPE$, then Construction \ref{construction:subset_martingale} produces a uniform family of $\GapP$-martingales.
	\end{enumerateC}
\end{lemma}
\begin{appendixproof}[Proof of Lemma \ref{le:subset_martingale_uniform_family}]
	If $B \in \UE$, then we claim $A \in \UP$. Given $w \in \{0,1\}^n$, for each $i < n$, if $w[i] = 1$, guess a witness for $s_i \in B$. If all witnesses are found, accept $w$. Because $A$ has unique witnesses, $B$ also has unique witnesses. Because the $s_i$'s have length logarithmic in the length of $w$, the total length of the witness for $w \in B$ is at most $|w|2^{O(\log |w|)} = |w|^{O(1)}$.
	Therefore Construction \ref{construction:cover_martingale} produces a uniform family of $\SharpP$-martingales.

	The other two cases follow similarly.
\end{appendixproof}

\begin{lemma}\label{le:subset_martingale_infinite_subsets}
	Let $B \subseteq \{0,1\}^*$ and assume the series  $\sum\limits_{n=0}^\infty  2^{c_B(n)-n}$ is $\p$-convergent.
	\begin{enumerateC}
		\item If $B \in \UE$, then
		the class of all infinite subsets of $B$ has $\SharpP$-measure 0.
		\item If $B \in \NE$, then
		the class of all infinite subsets of $B$ has $\SpanP$-measure 0.
		\item If $B \in \SPE$, then
		the class of all infinite subsets of $B$ has $\GapP$-measure 0.
	\end{enumerateC}
\end{lemma}
\begin{proof}
	Combine Lemma \ref{le:subset_martingale_uniform_family} and the Counting Measure Borel-Cantelli Lemma (Lemma \ref{lemma:counting-measure-borel-cantelli}).
\end{proof}

We now apply the subset construction along with the law of large numbers for $\P$-random languages to conclude that counting random languages do not belong to particular complexity classes.

\begin{corollary}\label{co:SharpP_SpanP_random_UP_NP}
	\begin{enumerateC}
		\item Every $\SharpP$-random language is not in $\UE$.
		\item Every $\SpanP$-random language is not in $\NE$.
		\item Every $\GapP$-random language is not in $\SPE$.
	\end{enumerateC}
\end{corollary}
\begin{proof}
	Let $R$ be $\SharpP$-random. Since $R$ is also $\p$-random, it satisfies the law of large numbers \cite{Lutz:QSET} and $c_R(n) \leq (\frac{1}{2}+\epsilon)n$ for any $\epsilon > 0$ and all sufficiently large $n$. The previous lemma applies to contradict the $\SharpP$-randomness of $R$. The proofs for $\SpanP$-random and $\GapP$-random languages are analogous.
\end{proof}

\subsection{Acceptance Probability Construction}\label{sec:acceptance_probability_construction}
Determining the $\p$-measure of $\BPP$ is an open problem. Van Melkebeek \cite{vanM00} used the weak derandomization of $\BPP$ from the uniform hardness assumption $\BPP \neq \EXP$ \cite{ImpWig01}
to prove a zero-one law for the $\p$-measure of $\BPP$:
either $\mup(\BPP) = 0$ or $\BPP = \EXP$.
Since $\BPP = \EXP$ is equivalent to $\mu(\BPP\mid\EXP) = 1$ \cite{ReSiCa95}, it follows that determining the $\p$-measure of $\BPP$ is equivalent to resolving the $\BPP$ versus $\EXP$ problem. However, in this section we will show that $\BPP$ has $\SharpP$-measure 0. We will also show that $\BQP$ has $\GapP$-measure 0 by building on the work of Fortnow and Rogers \cite{ForRog99,Fortnow03:One} that $\BQP \subseteq \AWPP$. Both of these results will be proved using the following martingale construction that tracks acceptance probabilities of classical or quantum machines.

\begin{construction_named}{Acceptance Probability Martingale}\label{construction:probability_martingale_construction}
	Let $f : \{0,1\}^* \times \{0,1\} \to \N$ be a function such that for some function $q(n)$ and all $x \in \{0,1\}^*$, \[ f(x,0) + f(x,1) = 2^{q(|x|)}. \]
	For each $w \in \{0,1\}^*$, let $n=|w|$ and define
	\[ M(w) = \prod_{i=0}^{n-1} 2 f(s_i,w[i])
		= 2^n \prod_{i=0}^{n-1} f(s_i,w[i]), \]
	\[ N(w) =  \prod_{i=0}^{n-1} 2^{q(|s_i|)} = 2^{\sum\limits_{i=0}^{n-1} q(|s_i|)}, \]
	and \[ d(w) = \frac{M(w)}{N(w)} = 2^n \prod_{i=0}^{n-1} \frac{f(s_i,w[i])}{2^{q(|s_i|)}}. \]
	Then $d$ is a martingale with $d(\lambda) = 1$.
	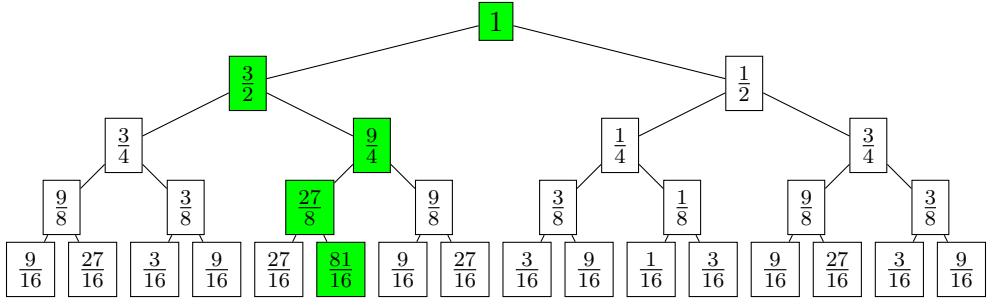
\begin{figure}
		\begin{center}
			\begin{tikzpicture}[level distance=1.5cm,
					level 1/.style={sibling distance=12cm},
					level 2/.style={sibling distance=6cm},
					level 3/.style={sibling distance=3cm},
					level 4/.style={sibling distance=1.5cm},scale=.55,
					every node/.style={draw, rectangle, align=center}]
				\node [fill=green] {$1$}
				child {node [fill=green] {$\tfrac{3}{2}$}
						child {node {$\tfrac{3}{4}$}
								child {node {$\tfrac{9}{8}$}
										child {node {$\tfrac{9}{16}$}}
										child {node {$\tfrac{27}{16}$}}
									}
								child {node {$\tfrac{3}{8}$}
										child {node {$\tfrac{3}{16}$}}
										child {node {$\tfrac{9}{16}$}}
									}
							}
						child {node [fill=green] {$\tfrac{9}{4}$}
								child {node [fill=green] {$\tfrac{27}{8}$}
										child {node {$\tfrac{27}{16}$}}
										child {node [fill=green] {$\tfrac{81}{16}$}}
									}
								child {node {$\tfrac{9}{8}$}
										child {node {$\tfrac{9}{16}$}}
										child {node {$\tfrac{27}{16}$}}
									}
							}
					}
				child {node {$\frac{1}{2}$}
						child {node {$\tfrac{1}{4}$}
								child {node {$\tfrac{3}{8}$}
										child {node {$\tfrac{3}{16}$}}
										child {node {$\tfrac{9}{16}$}}
									}
								child {node {$\tfrac{1}{8}$}
										child {node {$\tfrac{1}{16}$}}
										child {node {$\tfrac{3}{16}$}}
									}
							}
						child {node {$\tfrac{3}{4}$}
								child {node {$\tfrac{9}{8}$}
										child {node {$\tfrac{9}{16}$}}
										child {node {$\tfrac{27}{16}$}}
									}
								child {node {$\tfrac{3}{8}$}
										child {node {$\tfrac{3}{16}$}}
										child {node {$\tfrac{9}{16}$}}
									}
							}
					};
			\end{tikzpicture}
		\end{center}\caption{Acceptance Probability Martingale Construction (Construction \ref{construction:probability_martingale_construction})}\label{figure:acceptance_probability_martingale}
	\end{figure} See Figure \ref{figure:acceptance_probability_martingale} for an example with $n = 4$, $A \cap [s_0,s_3] = \{s_1,s_3\}$, and correctness probability $\tfrac{3}{4}$. The green path highlights the prefix 0101 of the characteristic sequence of $A$.
\end{construction_named}

For example, the function $f$ in Construction \ref{construction:probability_martingale_construction} could be the $\SharpP$ function where $f(x,0)$ is the number of rejecting paths and $f(x,1)$ is the number of accepting paths in a $\PTM$ $Q$ on input $x$. Let $q(n)$ be the random seed length of $Q$. Then
\[ \frac{f(x,0)}{2^{q(|x|)}} = \Pr[Q\textrm{ rejects }x] = \Pr[Q(x) = 0], \]
\[ \frac{f(x,1)}{2^{q(|x|)}} = \Pr[Q\textrm{ accepts }x] = \Pr[Q(x) = 1], \]
and
\[ d(w) =2^n \prod_{i=0}^{n-1} \Pr[Q(s_i) = w[i]]. \]
In particular, if $A \in \BPE = \BPTIME(2^{O(n)})$, then for some $c \geq 1$, there exists a $2^{cn}$-time $\PTM$ $Q$ that decides $A$ with error probability at most $2^{-2n}$,  where the random seed length is $q(n) = 2^{cn}$.

\begin{lemma} \label{le:bptime_martingales}
	If $A \in \BPE$, then the martingale $d$ produced by Construction \ref{construction:probability_martingale_construction} is a $\SharpP$-martingale that $0$-strongly succeeds on $A$.
\end{lemma}
\begin{appendixproof}[Proof of Lemma \ref{le:bptime_martingales}]
	Let $t(n) = 2^{cn}$ and $p(n)=2n$. Then
	$M(w)$ is $\SharpP$-computable with run time on the order of
	\[ \sum_{i=0}^{n-1} t(|s_i|) \leq n \cdot t(|s_n|) \leq n 2^{c\log n} = n^{c+1}. \]
	Similarly,
	$N(w)$ is computable in $O(n^{c+1})$ time. Therefore $d$ is a $\SharpP$-martingale.
	We have
	\[ d(A \restr n) = \frac{M(A\restr n)}{N(A \restr n)} = 2^n \prod_{i=0}^{n-1} \Pr[M(s_i) = A[i]] \geq 2^n \prod_{i=0}^{n-1} (1-2^{-p(|s_i|)}) \]
	Using $1 - x \approx e^{-x}$, we have $d(A\restr n) = \Omega(2^n)$ because
	\[ 2^n \prod_{i=0}^{n-1} e^{-2^{-p(|s_i|)}} = 2^n e^{-\sum\limits_{i=0}^{n-1} 2^{-p(|s_i|)}} = \Omega(2^n). \]
	The last line holds because
	\[ \sum_{i=0}^\infty 2^{-p(|s_i|)} = \sum_{n=0}^\infty 2^n \cdot 2^{-2n} \] converges. Therefore $d$ $0$-strongly succeeds on $A$.
\end{appendixproof}

Analogously, one can handle $\BQE = \BQTIME(2^{O(n)})$ \cite{ArunachalamGriloGurOliveiraSundaram02,ChiaChouZhangZhang:ITCS22}
with $\GapP$ functions. For this, we need the class $\AWPP$ that was introduced by Fenner, Fortnow, Kurtz, and Li \cite{FennerFortnowKurtzLi03}.
The following definition is due to Fenner \cite{Fenner03}. See \cite{FeFoKu94,ForRog99,Fenner03} for more details.
\begin{definition}
	The class $\AWPP$ (almost-wide probabilistic polynomial-time) consists of
	the languages $L$ such that for all polynomials $r$, there is a polynomial $t$ and a $\GapP$
	function $g$ such that, for all $n$, for all $x \in \binary^n$,
	\begin{itemize}
		\item if $x \in L$, then $1 - 2^{-r(n)}\leq \frac{g(x)}{2^{t(n)}} \leq 1$, and
		\item if $x \not\in L$, then $0 \leq \frac{g(x)}{2^{t(n)}} \leq 2^{-r(n)}$.
	\end{itemize}
\end{definition}

We define an exponential version $\AWPE$ by allowing $g$ to be computable in time $2^{O(n)}$ in the definition above. The class $\GapE$ is defined just like $\GapP$ but using $\PTM$s that run in $2^{O(n)}$ time.

\begin{definition}
	The class $\AWPE$ (almost-wide probabilistic exponential-time) consists of
	the languages $L$ such that for all $r(n) = 2^{O(n)}$, there is $t(n)=2^{O(n)}$ and a $\GapE$
	function $g$ such that, for all $n$, for all $x \in \binary^n$,
	\begin{itemize}
		\item if $x \in L$, then $1 - 2^{-r(n)}\leq \frac{g(x)}{2^{t(n)}} \leq 1$, and
		\item if $x \not\in L$, then $0 \leq \frac{g(x)}{2^{t(n)}} \leq 2^{-r(n)}$.
	\end{itemize}
\end{definition}

\begin{theorem}\label{th:AWPTIME_GapP_1_succeeds}
	If $A \in \AWPE$, then Construction \ref{construction:probability_martingale_construction} is a $\GapP$-martingale that $0$-strongly succeeds on $A$.
\end{theorem}
\begin{appendixproof}[Proof of Theorem \ref{th:AWPTIME_GapP_1_succeeds}]
	Assume $A \in \AWPE$, pick $r(n) = 2^n$, and let $t(n) = 2^{cn}$ and $g \in \GapE$ be the corresponding functions from the definition of $\AWPE$ such that for all $n$ and for all $x \in \binary^n$,
	\begin{itemize}
		\item if $x \in A$, then $1 - 2^{-r(n)}\leq \frac{g(x)}{2^{t(n)}} \leq 1$,
		\item if $x \notin A$, then $0 \leq \frac{g(x)}{2^{t(n)}} \leq 2^{-r(n)}$.
	\end{itemize}
	To adapt the acceptance probability construction, define $f(x, 0) = 2^{t(n)} - g(x)$ and $f(x, 1) = g(x)$, and define
	\[M(w) = \prod_{i=0}^{n-1} 2 f(s_i,w[i])
		= 2^n \prod_{i=0}^{n-1} f(s_i,w[i]).\]
	An argument similar to the proof of Lemma \ref{le:bptime_martingales} together with the closure properties of $\GapP$ \cite{FeFoKu94} shows that $M$ computes a $\GapP$ function.
	Also define
	\[N(w) =  \prod_{i=0}^{n-1} 2^{t(|w_i|)}\]
	Now we have

	\[ d(A \restr n) = \frac{M(A\restr n)}{N(A \restr n)} = 2^n \prod_{i=0}^{n-1} \frac{f(s_i, w[i])}{2^{t(|w_i|)}} \geq 2^n \prod_{i=0}^{n-1} (1-2^{-r(|s_i|)}) \]
	Using $1 - x \approx e^{-x}$, this yields $d(A\restr n) = \Omega(2^n)$ because
	\[2^n \prod_{i=0}^{n-1} e^{-2^{-r(|s_i|)}} = 2^n e^{-\sum\limits_{i=0}^{n-1} 2^{-r(|s_i|)}} = \Omega(2^n).\]
	Therefore $d$ $0$-strongly succeeds on $A$.
\end{appendixproof}
A straightforward extension of Fortnow and Rogers' proof that $\BQP \subseteq \AWPP$ to exponential-time classes shows that
$\BQE \subseteq \AWPE$.
Combining this with Theorem \ref{th:AWPTIME_GapP_1_succeeds} gives us the following result.
\begin{theorem}\label{th:BQE_BQEXP_1_succeeds}
	If $A \in \BQE$, then there is a $\GapP$-martingale that $0$-strongly succeeds on $A$.
\end{theorem}

Athreya et al. \cite{Athreya:ESDAICC} showed that $\DTIME(2^{cn})$ has $\P$-strong dimension 0 for all $c \geq 1$.
Using the above results and the Counting Measure Union lemma \ref{lemma:counting-measure-union}, we can extend this result to $\BPTIME$ and $\BQTIME$ under $\SharpP$ and $\GapP$ strong dimensions.

\begin{corollary}\label{co:BPTIME_BQTIME_dimension_0}
	For all $c \geq 1$,
	\begin{enumerateC}
		\item $\BPTIME(2^{cn})$ has $\SharpP$-strong dimension 0.
		\item $\BQTIME(2^{cn})$ has $\GapP$-strong dimension 0.
	\end{enumerateC}
\end{corollary}

\begin{corollary}\label{co:BPP_BQP_dimension_0}
	\begin{enumerateC}
		\item    $\BPP$ has $\SharpP$-strong dimension 0.
		\item    $\BQP$ has $\GapP$-strong dimension 0.
	\end{enumerateC}
\end{corollary}

We have the following Corollary to this section as a companion to Corollary \ref{co:SharpP_SpanP_random_UP_NP}. \begin{corollary}\label{co:random_oracles_BPE_BPEXP_BQE_BQEXP}
	\begin{enumerateC}
		\item Every $\SharpP$-random oracle is not in $\BPE$.
		\item Every $\GapP$-random oracle is not in $\BQE$.
	\end{enumerateC}
\end{corollary}

\subsection{Bi-immunity Martingale Construction}

Mayordomo \cite{Mayo94} showed that $\P$-random languages are $\E$-bi-immune. We extend her construction to show that $\SharpP$-random languages are $\UE\cap\coUE$-immune.
The same construction also shows $\SpanP$-random languages are $\NE\cap\coNE$-bi-immune and $\GapP$-random languages are $\SPE$-bi-immune.

\newcommand{\Lw}{L(w)}

\begin{construction_named}{Bi-immunity Martingale}\label{construction:bi-immunity_martingale}
	Let $A \subseteq \{0,1\}^*$. Define a martingale $d$ by $d(\lambda) = 1$ and for all $w \in \{0,1\}^*$,\
	\[ d(w1) =
		\begin{cases}
			2d(w) & \textrm{if } s_{|w|} \in A     \\
			d(w)  & \textrm{if } s_{|w|} \not\in A
		\end{cases} \]
	and
	\[ d(w0) =
		\begin{cases}
			0    & \textrm{if } s_{|w|} \in A      \\
			d(w) & \textrm{if } s_{|w|} \not\in A.
		\end{cases} \]
	For all $w \in \{0,1\}^*$, again writing $\Lw = \{ s_i \mid w[i] = 1 \},$
	we have
	\[ d(w) =
		\begin{cases}
			2^{\#_1(A\thinspace\restr\thinspace n)} & \textrm{if }A \subseteq L(w) \\
			0                                       & \textrm{otherwise.}
		\end{cases} \]

	\begin{figure}
		\begin{center}
			\begin{tikzpicture}[level distance=1.5cm,
					level 1/.style={sibling distance=12cm},
					level 2/.style={sibling distance=6cm},
					level 3/.style={sibling distance=3cm},
					level 4/.style={sibling distance=1.5cm},scale=.47,
					every node/.style={draw, rectangle, align=center}]
				\node[fill=green]{1}
				child {node [fill=green] {1}
						child {node  {0}
								child {node  {0}
										child {node  {0}} child {node  {0}} }
								child {node {0}
										child {node {0}} child {node {0}} }
							}
						child {node [fill=green] {2}
								child {node [fill=green] {2}
										child {node {0} } child {node [fill=green] {4}} }
								child {node [fill=green] {2}
										child {node {0}}
										child {node [fill=green] {4}}
									}
							}
					}
				child {node [fill=green] {1}
						child {node {0}
								child {node {0}
										child {node {0}}
										child {node {0}}
									}
								child {node {0}
										child {node {0}}
										child {node {0}}
									}
							}
						child {node [fill=green] {2}
								child {node [fill=green] {2}
										child {node {0}}
										child {node [fill=green] {4}}
									}
								child {node [fill=green] {2}
										child {node {0}}
										child {node [fill=green] {4}}
									}
							}
					};
			\end{tikzpicture}
		\end{center}\caption{Bi-immunity Martingale Construction (Construction \ref{construction:bi-immunity_martingale})}\label{figure:biimunity_martingale}
	\end{figure}
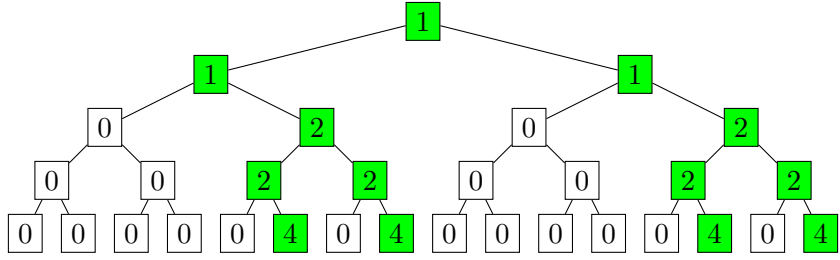 See Figure \ref{figure:biimunity_martingale} for an example with $n = 4$ and $A \cap [s_0,s_3] = \{s_1,s_3\}$. The nodes $w \in \{0,1\}^{\leq 4}$ that are colored green in the tree are those with $A \subseteq L(w)$.

\end{construction_named}
Mayordomo \cite{Mayo94}  showed that if $A \in \E$, then Construction \ref{construction:bi-immunity_martingale} is a $\P$-martingale and if $A \in \ESPACE$, then Construction \ref{construction:bi-immunity_martingale} is a $\PSPACE$-martingale.
For any infinite language $A$, Mayordomo showed Construction \ref{construction:bi-immunity_martingale} succeeds on all languages that contain $A$.

\begin{lemma_cite}{Mayordomo \cite{Mayo94}}
	If $A$ is infinite, then Construction \ref{construction:bi-immunity_martingale} has $S^\infty[d] = \{ B \mid A \subseteq B \}$.
\end{lemma_cite}

We generalize Mayordomo's construction to counting martingales.

\begin{lemma}\label{le:bi-immunity_martingale_spp_gapp}
	\begin{enumerateC}
		\item If $A \in \UE \cap \coUE$, then Construction \ref{construction:bi-immunity_martingale} is an exact $\SharpP$-martingale.\item If $A \in \NE \cap \coNE$, then Construction \ref{construction:bi-immunity_martingale} is an exact $\SpanP$-martingale.
		\item If $A \in \SPE$, then Construction \ref{construction:bi-immunity_martingale} is an exact $\GapP$-martingale.
	\end{enumerateC}
\end{lemma}
\begin{appendixproof}[Proof of Lemma \ref{le:bi-immunity_martingale_spp_gapp}]
	For parts 1 and 2, on input $w$, for each $i < |w|$, guess whether $s_i \in A$ or $s_i \not\in A$ and a witness that proves this.
	Note that the witnesses have total length $2^{O(n)}$ which is polynomial in $|w|$. If witnesses are found for all $s_i$ and prove that $A \restr |w| \subseteq \Lw$, output $d(w) = 2^{|A
				\;\restr \; |w||}$; otherwise output $0$.
	\begin{enumerate}
		\item If $A \in \UE \cap \coUE$, then $d$ is a $\UPSV$ function which may be implemented in $\SharpP$.
		\item If $A \in \NE \cap \coNE$, then $d$ is an $\NPSV$ function which may be implemented in $\SpanP$.\end{enumerate}
	For part 3, let $g \in \GapE$ such that $g(x) = 1$ if $x \in A$ and $g(x) = 0$ if $x \not \in A$. Define
	\begin{eqnarray*}
		f(w,1) &=& 1 + g(s_{|w|})\\
		f(w,0) &=& 1 - g(s_{|w|})
	\end{eqnarray*}
	for all $w \in \{0,1\}^*$. Then $f \in \GapP$ and \[ d(wb) = f(w,b) d(w) \] for all $w \in \{0,1\}^*$ and $b \in \{0,1\}$. Therefore
	\[ d(w) = \prod_{i=0}^{|w|-1} f(w\restr i, w[i]) \]
	is a $\GapP$ function.
\end{appendixproof}

We can now conclude bi-immunity results for counting random languages. Part 3 of the following corollary improves part 3 of Corollary \ref{co:SharpP_SpanP_random_UP_NP}.

\begin{corollary}\label{co:gapp_random_spp_bi_immune}
	\begin{enumerateC}
		\item Every $\SharpP$-random language is $\UE \cap \coUE$-bi-immune.
		\item Every $\SpanP$-random language is $\NE \cap \coNE$-bi-immune.
		\item Every $\GapP$-random language is $\SPE$-bi-immune.
	\end{enumerateC}
\end{corollary}
\begin{appendixproof}[Proof of Corollary \ref{co:gapp_random_spp_bi_immune}]
	Let $R$ be $\GapP$-random and let $A \in \SPE$. By the previous two lemmas, $R$ is $A$-immune. Since $\SPE$ is closed under complement, $R^c$ is also $A$-immune. The other two parts are analogous.
\end{appendixproof}
Since $\SPP$ contains the counting classes $\UP$, $\FewP$, and $\Few$, $\GapP$-random languages are also immune to these classes.

\sectionnewpage
\section{Entropy Rates and Kolmogorov Complexity}\label{sec:entropy_rates_kolmogorov_complexity}
We will use the Cover Martingale and Conditional Expectation Martingale Constructions (Constructions \ref{construction:cover_martingale} and \ref{construction:random_variable_martingale}) to develop a few tools for working with counting measures and dimensions. First, we extend the entropy rates used by Hitchcock and Vinodchandran \cite{Hitchcock:DERC} to our setting. Then we extend Lutz's results on Kolmogorov complexity and $\PSPACE$-measure to counting measure.
\subsection{Entropy Rates}\label{subsec:entropy_rates}
We will show in this section that the Cover Martingale Construction (Construction \ref{construction:cover_martingale}) may be combined with the concept of entropy rates to build counting martingales.

\begin{definition}
	The {\em entropy
	rate} of a language $A \subseteq \{0,1\}^*$
	is \[ H_A = \limsupn \frac{\log |A_{=n}|}{n}. \]
\end{definition}
Intuitively, $H_A$ gives an asymptotic measurement of the amount by
which every string in $A_{=n}$ is compressed in an optimal code \cite{Kuich70}.
\begin{definition}  Let $A \subseteq \{0,1\}^*$.
	The {\em i.o.-class
			of $A$} is \[ \ioA = \{ S \in \C \mid (\exists^\infty n)\ S \restr
		n \in A\}. \]
	The {\em a.e.-class of $A$} is \[ \aeA = \{ S \in \C \mid (\forall^\infty n)\ S
		\restr n \in A\}. \]
\end{definition}
\noindent
That is, $\ioA$ is the class of sequences that have infinitely many
prefixes in $A$ and $\aeA$ is the class of sequences that have all but finitely many prefixes in $A$.

\newcommand{\HCstr}{\calH_\calC^\str}
\begin{definition_cite}{Hitchcock \cite{Hitchcock:CPED,Hitchcock:phdthesis}}  Let $\calC$ be a class of languages and $X
		\subseteq \C$.
	The {\em $\calC$-entropy rate} of $X$ is
	\[ \HC(X) = \inf \{ H_A \mid A \in \calC\textrm{ and }X \subseteq
		\ioA\}. \]
	The {\em $\calC$-strong entropy rate} of $X$ is
	\[ \HCstr(X) = \inf \{ H_A \mid A \in \calC\textrm{ and }X \subseteq
		\aeA\}. \]
\end{definition_cite}
\noindent
Informally, $\HC(X)$ is the lowest entropy rate with which every
element of $X$ can be covered infinitely often by a language in
$\calC$.
We may also interpret $\HC(X)$ as a notion of dimension.
For all $X \subseteq \C$, it is known that
$\dimh(X) = \calH_\ALL(X),$ where $\ALL$ is the class of all
languages and $\dimh$ is Hausdorff dimension (see \cite{Roge98,Staiger93}). Hitchcock \cite{Hitchcock:CPED,Hitchcock:phdthesis} showed that $\dimpack(X) = \calH_\ALL^\str(X)$ and using other
classes gives equivalent definitions of the constructive dimension ($\cdim(X) = \calH_\CE(X)$ and $\cDim(X) = \calH_\CE^\str$), computable dimension ($\dimcomp(X) = \calH_\DEC(X)$ and $\Dimcomp(X) = \calH_\DEC^\str(X)$), and polynomial-space dimension ($\dimpspace(X) = \calH_\PSPACE(X)$ and $\Dimpspace(X) = \calH_\PSPACE^\str(X)$).
For time-bounded dimension, no analogous result is known.
However, the following upper bounds are true \cite{Hitchcock:CPED,Hitchcock:phdthesis}:
for all $X \subseteq \C$, $\calH_{\P}(X) \leq \dimp(X)$ and $\calH_{\P}^\str(x) \leq \Dimp(X)$, where $\dimp$ is the polynomial-time dimension and $\Dimp$ is the polynomial-time strong dimension.

The $\NP$-entropy rate is an upper bound for
$\Deltapthree$-dimension.

\begin{theorem_cite}{Hitchcock and Vinodchandran \cite{Hitchcock:DERC}}\label{th:DERC_deltathree_HNP}
	For all $X \subseteq \C$,  \[ \dimDeltapthree(X)
		\leq \calH_\NP(X). \]
\end{theorem_cite}

The strong dimension analogue of Theorem \ref{th:DERC_deltathree_HNP} also holds:  \[ \Dim_\Deltapthree(X)  \leq \calH_\NP^\str(X) \] for all $X \subseteq \C$.

We now show that the $\NP$-entropy rate upper bounds $\SpanP$-dimension. Analogously, the $\UP$-entropy rate upper bounds $\SharpP$-dimension and the $\SPP$-entropy rate upper bounds $\GapP$-dimension. The proof uses Construction \ref{construction:cover_martingale} and Lemma \ref{le:cover_martingale_construction}.

\begin{theorem}\label{th:CountingDimLeqEntropyRates}
	For all $X \subseteq \C$,
	\begin{enumerateC}
		\item $\dim_\SharpP(X) \leq \calH_\UP(X) \leq \calH_\P(X)$
		\item $\dim_\SpanP(X) \leq \calH_\NP(X).$
		\item $\dim_\GapP(X) \leq \calH_\SPP(X).$
		\item $\Dim_\SharpP(X) \leq \calH_\UP^\str(X) \leq \calH_\P^\str(X)$
		\item $\Dim_\SpanP(X) \leq \calH_\NP^\str(X).$
		\item $\Dim_\GapP(X) \leq \calH_\SPP^\str(X).$
	\end{enumerateC}
\end{theorem}
\begin{appendixproof}[Proof of Theorem \ref{th:CountingDimLeqEntropyRates}]
	We prove the second inequality, the proofs of the other inequalities are analogous.

	Let $\alpha > \calH_\NP(X)$ and $\epsilon > 0$ such that $2^\alpha$ and
	$2^\epsilon$ are rational.  Let $A \in \NP$ such that $X \subseteq
		\ioA$ and $H_A < \alpha$.  We can assume that $|A_{=n}| \leq 2^{\alpha
				n}$ for all $n$.  It suffices to show that $\dimSpanP(X) \leq \alpha+\epsilon$.

	We use Construction \ref{construction:cover_martingale} and Lemma \ref{le:cover_martingale_construction} to obtain a uniform family $(d_n \mid n\geq 0)$ of $\SpanP$ martingales with
	\[ d_n(\lambda) = \frac{|A_{=n}|}{2^n} \leq 2^{(\alpha-1)n} \] and $d_n(v) = 1$ for all $v \in A_{=n}$. We now apply the Counting Dimension Borel-Cantelli Lemma (Lemma \ref{lemma:counting-dimension-borel-cantelli}) to complete the proof.
\end{appendixproof}

We note that combining Theorems \ref{th:CountingDimLeqEntropyRates} and \ref{th:counting_measure_conservation_deltaEthree} gives a new proof of Theorem \ref{th:DERC_deltathree_HNP}.

We will next extend Theorem \ref{th:CountingDimLeqEntropyRates} to the measure setting. First, we define an entropy rate version of measure 0. The idea in this definition is to extend the entropy rate $\calH_\calC$ to define a measure $\MC$.

\begin{definition_named}{Entropy Rate Measure}
	Let $X \subseteq \C$ and let $\calC$ be a complexity class.
	If there exist $A \in \calC$ and $f \in \FP$ such that
	\begin{enumerateC}
		\item $X \subseteq \ioA$,
		\item $\log |A_{=n}| < n - f(n)$ for sufficiently large $n$, and
		\item $\sum\limits_{n=0}^\infty 2^{-f(n)}$ is $\p$-convergent,
	\end{enumerateC}
	then $X$ has $\MC$-measure 0 and we write $\MC(X) = 0$.
\end{definition_named}

We observe that $\MC$-measure has many of the standard measure properties.  When using $\calC= \ALL$, the class of all languages, it refines Lebesgue measure: if $\ClassMeasure{\ALL}(X) = 0$, then $X$ has Lebesgue measure 0. Using $\calC=\PSPACE$, we have if $\ClassMeasure{\PSPACE}(X) = 0$, then $\mupspace(X) = 0$.

\begin{proposition}\label{prop:entropy_basic}
	Let $\calC, \calD$ be classes of languages and
	$X, Y \subseteq \C$.
	\begin{enumerateC}
		\item If $\calC \subseteq \calD$, $\MC(X) = 0$ implies $\ClassMeasure{\calD}(X) = 0$.
		\item If $X \subseteq Y$, then $\MC(Y) =0 $ implies $\MC(X) = 0$.
		\item If $\calC$ is closed under union, then $\MC(X) = 0$ and $\MC(Y) = 0$ implies $\MC(X \cup Y) = 0$.
		\item If $\HC(X) < 1$, then $\MC(X) = 0$.
	\end{enumerateC}
\end{proposition}

We now establish our measure-theoretic extension of Theorem \ref{th:CountingDimLeqEntropyRates}. This proof also uses Construction \ref{construction:cover_martingale} and Lemma \ref{le:cover_martingale_construction}.

\begin{theorem}\label{th:NP-entropy-SpanP-Measure}
	Let $X \subseteq \C$.
	\begin{enumerateC}
		\item      If $\ClassMeasure{\UP}(X) = 0$, then $\muSharpP(X) = 0$.
		\item      If $\MNP(X) = 0$, then $\muSpanP(X) = 0$.
		\item  If $\ClassMeasure{\SPP}(X) = 0$, then $\muGapP(X) = 0$.
	\end{enumerateC}
\end{theorem}
\begin{appendixproof}[Proof of Theorem \ref{th:NP-entropy-SpanP-Measure}]
	Assume $\MNP(X)=0$ and obtain the cover $A \in \NP$ and the function $f \in \FP$ satisfying the definition of $\MNP(X) = 0$.
	We use Construction \ref{construction:cover_martingale} and Lemma \ref{le:cover_martingale_construction} to obtain a uniform family of exact $\SpanP$ martingales $(d_n \mid n \geq 0)$ with
	\[d_n(\lambda) = \frac{|A_{=n}|}{2^n} \leq \frac{2^{n-f(n)}}{2^n} = 2^{-f(n)}\] and
	$d_n(w) = 1$ for all $w \in A_{=n}$.
	The Counting Measure Borel-Cantelli lemma (Lemma \ref{lemma:counting-measure-borel-cantelli}) completes the proof. The proofs of the other items are analogous.
\end{appendixproof}

\subsection{Kolmogorov Complexity}\label{subsec:kolmogorov_complexity}

Lutz \cite{Lutz:AEHNC} showed that the space-bounded Kolmogorov complexity class
\[ \{ S \mid (\exists^\infty n)\, KS^p(S \restr n) < n- f(n) \} \]
has $\pspace$-measure 0 where $p$ is a polynomial, and $\sum\limits_{n=0}^\infty 2^{-f(n)}$ is $\p$-convergent. In other words, if $S$ is $\pspace$-random, then $KS^p(S \restr n) \geq n-f(n)$ a.e. For $\p$-random sequences, Lutz proved that the time-bounded Kolmogorov complexity $K^p(S\restr n) \geq c \log n$ a.e. for any polynomial $p$.

We use the Conditional Expectation Martingale Construction (Construction \ref{construction:random_variable_martingale}) to obtain an intermediate result for time-bounded Kolmogorov complexity and $\SharpP$-measure.

\begin{theorem}\label{th:Kolmogorov_SharpP_measure}
	Suppose $f \in \FP$ and the series $\sum\limits_{n=0}^\infty 2^{-f(n)}$ is $\p$-convergent.
	Let $p$ be a polynomial and
	\[  X = \{S \mid (\exists^\infty n)\, K^p(S \restr n) < n- f(n) \}. \]
	Then $X$ has $\SharpP$-measure 0.
\end{theorem}
\begin{appendixproof}[Proof of Theorem \ref{th:Kolmogorov_SharpP_measure}]
	For each $n$, let
	\[ X_n = \{ x \in \{0,1\}^n \mid K^p(x) < n - f(n) \}. \]
	For $w \in \{0,1\}^{\leq n}$, let
	\[ d_n(w) = \Pr(X_n \mid w) = \frac{|\{ x \in X_n \mid w \prefix x \}|}{2^{n-|w|}}. \]
	This martingale satisfies \[ d_n(\lambda) \leq \frac{|X_n|}{2^n} \leq \frac{2^{n-f(n)}}{2^n} =  2^{-f(n)} \] and $d_n(x) = 1$ for all $x \in X_n$. However, $d_n$ does not appear to be a $\SharpP$ martingale because its numerator is a $\SpanP$ function.
	We will upper bound the $\SpanP$-martingale $d_n$ by another martingale $d_n'$ that is $\SharpP$-computable and still satisfies $d_n'(\lambda) \leq 2^{-f(n)}.$

	Let $U$ be a universal Turing machine and define
	\[ C = \myset{ \pair{0^n,w,x,\pi} }
		{\begin{array}{l}
				w \in \{0,1\}^{\leq n}, x \in \{0,1\}^n, \pi \in \{0,1\}^{<n-f(n)}, \\
				w \prefix x, \textrm{ and } U(\pi) = x\textrm{ in}\leq p(n)\textrm{ time}
			\end{array} }. \]
	Then $C \in \P$, so the function
	\[ g(0^n,w) = \left| \myset{ \pair{x,\pi} }{
			\begin{array}{l}
				x \in \{0,1\}^n, \pi \in \{0,1\}^{<n-f(n)}, w \prefix x \\
				\textrm{and } U(\pi) = x\textrm{ in}\leq p(n)\textrm{ time}
			\end{array} } \right|. \]
	is in $\SharpP$.
	We use Construction \ref{construction:random_variable_martingale}.
	Define the $\SharpP$-martingale
	\[ d_n'(w) = \frac{g(0^n,w)}{2^{n-|w|}} \]
	for all $w \in \{0,1\}^{\leq n}$ and $d_n'(y) = d_n'(y\restr n)$ for $y \in \{0,1\}^{>n}.$
	Notice that $g(0^n,w) \leq 2^{n-f(n)}$, so $d_n'(\lambda) \leq 2^{-f(n)}$.
	Also, $d_n'(x) \geq d_n(x)$ for all $x \in \{0,1\}^*$.
	If $x \in X_n$, then
	\[ g(0^n,x) =
		\left| \myset{ \pi }{
			\begin{array}{l}
				x \in \{0,1\}^n, \pi \in \{0,1\}^{<n-f(n)}, \\
				\textrm{and } U(\pi) = x\textrm{ in }\leq p(n)\textrm{ time}
			\end{array} } \right| \geq 1. \]
	and $d_n'(x) \geq 1 = d_n(x).$
	Therefore
	$X_n
		\subseteq S^1[d_n'].$ Also, $(d_n'\mid n \in \N)$ is exactly and uniformly $\SharpP$-computable by Lemma \ref{le:random_variable_martingale_construction}.
	We apply the Counting Measure
	Borel-Cantelli Lemma (Lemma \ref{lemma:counting-measure-borel-cantelli})
	to conclude that \[ X \subseteq \bigcap_{i=0}^\infty \bigcup_{j \geq i}^\infty S^1[d_n'] \]
	has $\SharpP$-measure 0.
\end{appendixproof}

\begin{corollary}
	Suppose $f \in \FP$ and the series $\sum\limits_{n=0}^\infty 2^{-f(n)}$ is $\p$-convergent.
	If $S$ is $\SharpP$-random, then $K^p(S \restr n) \geq n - f(n)$ a.e.
\end{corollary}

The following Theorem is a variation of Theorem \ref{th:Kolmogorov_SharpP_measure} which will be used in Section \ref{sec:applications}.

\begin{theorem}\label{th:Kolmogorov_SharpP_measure_second_version}
	Suppose $f \in \FP$ and the series $\sum\limits_{n=0}^\infty 2^{-f(n)}$ is $\p$-convergent.
	Let $p$ be a polynomial and
	\[  X = \{ A \mid (\exists^\infty n) \,K^p(A_{=n}) < 2^n- f(2^n) \}. \]
	Then $X$ has $\SharpP$-measure 0.
\end{theorem}
\begin{appendixproof}[Proof of Theorem \ref{th:Kolmogorov_SharpP_measure_second_version}]
	If $K^p(A_{=n}) < 2^n- f(2^n)$ then we have
	\begin{eqnarray*}
		K^p(A_{\leq n}) & \leq & K^p(A_{< n}) + K^p(A_{= n}) + O(n) \\
		& \leq     &         2^n + 2^n - f(2^n) + O(n) \\
		&  = & 2^{n+1} - f(2^n) + O(n).
	\end{eqnarray*}
	It follows from Theorem \ref{th:Kolmogorov_SharpP_measure} that $X$ has $\SharpP$-measure 0.
\end{appendixproof}

We next consider Kolmogorov complexity rates, which leads to an interesting connection with one-way functions.

\begin{definition_cite}{\cite{Hitchcock:phdthesis,Hitchcock:DERC}} Let $X \subseteq \C$.
	\begin{enumerate}
		\item The {\em polynomial-time Kolmogorov complexity rate of $X$} is
		      \[ \calK_\poly(X) = \inf_{p \in \poly} \sup_{S \in X} \liminfn \frac{K^p(S\restr n)}{n}. \]
		\item The {\em polynomial-time strong Kolmogorov complexity rate of $X$} is
		      \[ \calK_\poly^\str(X) = \inf_{p \in \poly} \sup_{S \in X} \limsupn \frac{K^p(S\restr n)}{n}. \]
	\end{enumerate}
	\begin{enumerate}
		\item The {\em polynomial-space Kolmogorov complexity rate of $X$} is
		      \[ \calKS_\poly(X) = \inf_{p \in \poly} \sup_{S \in X} \liminfn \frac{KS^p(S\restr n)}{n}. \]
		\item The {\em polynomial-space strong Kolmogorov complexity rate of $X$} is
		      \[ \calKS_\poly^\str(X) = \inf_{p \in \poly} \sup_{S \in X} \limsupn \frac{KS^p(S\restr n)}{n}. \]
	\end{enumerate}
\end{definition_cite}

Hitchcock and Vinodchandran \cite{Hitchcock:DERC} showed that for all $X \subseteq \C$,
\begin{equation}\label{eq:DERC_inequalities}
	\left\{ \begin{array}{c}
		\dimpspace(X) \\ = \\
		\HPSPACE(X)   \\ = \\
		\calKS^{\poly}(X)
	\end{array} \right\}
	\leq
	\dim_\Deltapthree(X) \leq
	\HNP(X) \leq
	\left\{ \begin{array}{c} \HP(X), \\ \calK_{\poly}(X) \end{array}\right\}
	\leq
	\dimp(X) .\end{equation}
At the polynomial-space level, $\PSPACE$-dimension, the $\PSPACE$-entropy rate, and the Kolmogorov complexity rate all coincide. At the polynomial-time level, the $\P$-dimension, $\P$-entropy rate, and time-bounded Kolmogorov complexity rate are not known to be equal. No relationship is known between $\HP(X)$ and $\calK_\poly(X)$.
Analogous inequalities hold for the strong dimension versions of the quantities in \eqref{eq:DERC_inequalities}: for all $X \subseteq \C$,
\begin{equation}\label{eq:DERC_inequalities_strong_version}
	\left\{ \begin{array}{c}
		\Dimpspace(X)         \\ = \\
		\calH^\str_\PSPACE(X) \\ = \\
		\calKS^\str_{\poly}(X)
	\end{array} \right\}
	\leq
	\Dim_\Deltapthree(X) \leq
	\calH^\str_\NP(X) \leq
	\left\{ \begin{array}{c} \calH^\str_\P(X), \\ \calK^\str_{\poly}(X) \end{array}\right\}
	\leq
	\Dimp(X) .\end{equation}
Each quantity in \eqref{eq:DERC_inequalities_strong_version} is greater than or equal to the corresponding quantity in \ref{eq:DERC_inequalities}.

We now show that the polynomial-time Kolmogorov complexity rates upper bound the $\SharpP$-dimensions.
\begin{theorem}\label{th:Kolmogorov_rate_SharpP_dimension}
	For all $X \subseteq \C$,
	\[ \dimSharpP(X) \leq \calK_\poly(X) \]
	and
	\[ \DimSharpP(X) \leq \calK_\poly^\str(X). \]
\end{theorem}
\begin{proof}
	Let $s > \calK^\poly(X)$ be rational.
	For each $n$, let \[ X_n =
		\{ x \in \{0,1\}^n \mid K^p(x) \leq s|x| \}. \]
	Then use Construction \ref{construction:random_variable_martingale} as in the proof of Theorem \ref{th:Kolmogorov_SharpP_measure} to obtain a $\SharpP$-martingale $d_n$ for each $n$ where $d_n(\lambda) \geq 2^{(s-1)n}$ for all $x \in X_n.$ We then apply the Counting Dimension Borel-Cantelli Lemma (Lemma \ref{lemma:counting-dimension-borel-cantelli}).
	This shows $\dimsharpp(X) \leq s$. The proof of the lower bounds $\dimsharpp(X) \geq s$ follows from \cite{Hitchcock:SDNC}. The proof for strong dimension is analogous.
\end{proof}

Combining Theorem \ref{th:CountingDimLeqEntropyRates}, Theorem \ref{th:Kolmogorov_rate_SharpP_dimension}, Corollary \ref{co:PtoSharpP_dominates_counting_measures_and_dimensions}, and the inequalities in \eqref{eq:DERC_inequalities}, we have the refined picture in Figure \ref{fig:kolmogorov_relationships_diagram}.
\begin{figure}
	\begin{center}
		\begin{tikzpicture}[
				node distance = 0.6cm,
				every node/.style={font=\footnotesize},
				dim/.style={draw, rectangle, rounded corners, fill=blue!20},
				calH/.style={draw, rectangle, rounded corners, fill=orange!40},
				calK/.style={draw, rectangle, rounded corners,fill=green!20},
				->, >=Stealth
			]

			\node[dim] (dimpspace) {$\dimpspace$};
			\node[dim, right=of dimpspace] (dimpsharpp) {$\dim_{\P^\SharpP}$};
			\node[dim, right=of dimpsharpp] (dimdeltapthree) {$\dimDeltaPthree$};
			\node[dim, right=of dimdeltapthree] (dimspanp) {$\dimSpanP$};
			\node[dim, right=of dimspanp] (dimsharpp) {$\dimSharpP$};

			\node[calH, right=of dimsharpp] (HUP) {$\calH_\UP$};
			\node[calH, right=of HUP] (HP) {$\HP$};
			\node[dim, right=of HP] (dimp) {$\dimp$};

			\node[calH, above=of dimsharpp] (HNP) {$\calH_\NP$};
			\node[calK, above=of HUP] (Kpoly) {$\calK^{\poly}$};

			\node[calH, below=of dimsharpp] (HSPP) {$\calH_\SPP$};
			\node[dim, below=of dimspanp] (dimgapp) {$\dimGapP$};

			\node[calK, above=of dimpspace] (KSpoly) {$\calKS^\poly$};
			\node[calH, below=of dimpspace] (HPSPACE) {$\calH_\PSPACE$};

			\draw (dimpspace) -- (dimpsharpp);
			\draw (dimpsharpp) -- (dimdeltapthree);
			\draw (dimdeltapthree) -- (dimspanp);
			\draw (dimspanp) -- (dimsharpp);
			\draw (dimsharpp) -- (HUP);

			\draw (HUP) -- (HP);
			\draw (HP) -- (dimp);

			\path[->] (dimpsharpp) edge [bend right=10] node {} (dimgapp);
			\draw (dimgapp) -- (HSPP);
			\path[->] (HSPP) edge [bend right=10] node {} (HUP);
			\path[->] (dimgapp) edge [bend right=10] node {} (dimsharpp);

			\path[->] (dimsharpp) edge [bend right=10] node {} (Kpoly);
			\path[->] (Kpoly) edge [bend left=10] node {} (dimp);

			\path[->] (dimspanp) edge [bend left=10] node {} (HNP);
			\path[->] (HNP) edge [bend left=10] node {} (HUP);

			\draw (HNP) -- (Kpoly);

			\draw[<->] (dimpspace) -- (KSpoly);
			\draw[<->] (dimpspace) -- (HPSPACE);

		\end{tikzpicture}

	\end{center}
	\caption{Relationships between Dimensions, Entropy Rates, and Kolmogorov Complexity Rates}\label{fig:kolmogorov_relationships_diagram}
\end{figure}
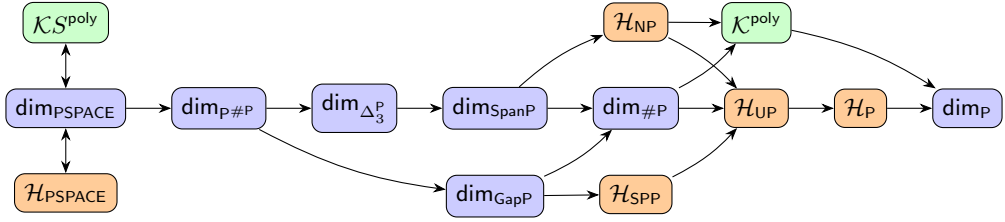
An arrow denotes that the dimension notion on the left is at most the dimension notion on the right. A double arrow denotes that the two dimension notions are equal.
Analogous inequalities hold for the strong dimension versions of the quantities in Figure \ref{fig:kolmogorov_relationships_diagram}

Nandakumar, Pulari, Akhil S, and Sarma \cite{NandakumarPulariSSarma24} showed that if one-way functions exist, then for all $\epsilon > 0$, there exists $X \subseteq \C$ with $\dimp(X) - \calK_\poly^\str(X) \geq 1-\epsilon$. In fact, $X$ may be taken as a singleton. Combining this result with Theorem \ref{th:Kolmogorov_rate_SharpP_dimension} and the inequalities $\calK_\poly(X) \leq \dimp(X)$ and $\calK^\str_\poly \leq \Dimp(X)$ from \eqref{eq:DERC_inequalities} and \eqref{eq:DERC_inequalities_strong_version}, we obtain the following corollaries.

\begin{corollary}
	If one-way functions exist, then for all $\epsilon > 0$, there exists $S \in  \C$ with $\dimp(S) - \dimSharpP(S) \geq 1-\epsilon$ and $\Dimp(S) - \DimSharpP(S) \geq 1-\epsilon$.
\end{corollary}

In other words: if one-way functions exist, then $\P$-dimension is different from $\SharpP$-dimension and strong $\P$-dimension is different from strong $\SharpP$-dimension.

\begin{corollary}
	\begin{enumerate}
		\item If $\dimSharpP(S) = \dimp(S)$ for all $S \in \C$, then one-way functions do not exist.
		\item If $\DimSharpP(S) = \Dimp(S)$ for all $S \in \C$, then one-way functions do not exist.
	\end{enumerate}
\end{corollary}

\sectionnewpage
\section{Applications}\label{sec:applications}
This section contains our main applications. We start with classical circuit complexity, then move on to quantum circuit complexity, and lastly the density of hard sets.
\subsection{Classical Circuit Complexity}
Lutz \cite{Lutz:AEHNC} showed that for all \(\alpha < 1\), the class \[X_\alpha = \SIZEio\left(\frac{2^n}{n}\left(1+\frac{\alpha \log n}{n}\right)\right)\] has \(\pspace\)-measure 0. Additionally, Lutz showed that for any \(c \geq 1\) and \(k \geq 1\), the classes \(\P/cn\) and \(\SIZE(n^k)\) have polynomial-time measure 0 and quasipolynomial-time measure 0, respectively.
Mayordomo \cite{Mayo94b} used Stockmeyer's approximate counting
of $\sharpP$ functions \cite{Stoc85} to show that $\P/\poly$ has
measure 0 in the third level of the exponential hierarchy.

We begin by extending Lutz's $\PSPACE$-measure result to $\MNP$-measure, in order to improve the theorem to $\SpanP$-measure. The proof uses the Minimum Circuit Size Problem ($\MCSP$) \cite{KabCai00} to form a cover. In $\MCSP$, we are given the full $2^n$-length truth-table of a Boolean function $f : \{0,1\}^n \to \{0,1\}$ and a number $s \geq 1$ and asked to decide whether there is a circuit of size at most $s$ computing $f$. The $\MCSP$ problem is in $\NP$ and not known to be $\NP$-complete \cite{KabCai00,MurrayWilliams17,Hitchcock:NPCMCSP}.
The $\MCSP$ problem fits perfectly into the $\MNP$ framework to help improve Lutz's result, that $X_\alpha$ has $\pspace$-measure 0.

\begin{theorem}\label{th:SIZEio_MNP_measure_0}
	For all $\alpha < 1$,
	\[\SIZEio\left(\frac{2^n}{n}\left(1+\frac{\alpha \log n}{n}\right)\right)\]
	has $\MNP$-measure 0.
\end{theorem}
\begin{appendixproof}[Proof of Theorem \ref{th:SIZEio_MNP_measure_0}]
	Let $s(n) = \frac{2^n}{n}\left(1+\frac{\alpha \log n}{n}\right)$ where $\alpha < 1$ and let $X = \SIZEio(s(n))$.

	Define
	\begin{eqnarray*}
		A
		&=& \myset{B_{\leq n} }{\pair{B_{=n},s(n)} \in \MCSP }\\
		&=& \myset{B_{\leq n} }{ B_{=n} \text{ has a circuit of size at most } s(n) }.
	\end{eqnarray*}
	Then $A \in \NP$, $X \subseteq \ioA$, and we need to show that $\log |A_{=N}| < N - f(N)$, where $N = 2^{n+1} - 1$,  for some $f \in \FP$ such that $\sum_{n=0}^\infty 2^{-f(n)}$ is $\p$-convergent. Using the bound in \cite{Lutz:AEHNC}, we have:
	\begin{eqnarray*}
		\log |A_{=N}| & < & \sum\limits_{m=0}^{n - 1} 2^m + \log \left( 48es\left(n\right)\right)^{s(n)} \\
		& = & 2^n - 1 + s(n)\left(\log (48e) + \log \left(s\left(n\right)\right)\right) \\
		& = & 2^n - 1 + \frac{2^n}{n}\left(1+\frac{\alpha \log n}{n}\right) \left( \log \left(48e\right) + n - \log n + \log \left( 1+\frac{\alpha \log n}{n}\right) \right) \\
		& < & 2^n - 1 + \frac{2^n}{n}\left(n + \alpha \log n - \log n + 6 \right) \\
		& < & 2^{n+1} - 1 - \frac{2^n}{n}\left(1 - \frac{\alpha}{2}\right)\log n \\
		& = & N -   \frac{2^n}{n}\left(1 - \frac{\alpha}{2}\right)\log n.
	\end{eqnarray*}
	As a result, set
	\begin{eqnarray*}
		f(N) & = & \left(1 - \frac{\alpha}{2}\right)\frac{2^n}{n}\log n \\
		& = & \left(1 - \frac{\alpha}{2}\right) \frac{N+1}{2(\log \left(N+1\right) - 1)} \log \left(\log \left(N + 1\right) - 1\right).
	\end{eqnarray*}
	It is easy to see that $f \in \FP$ and  $\sum_{n=0}^\infty 2^{-f(n)}$ is $\p$-convergent.
\end{appendixproof}

\begin{corollary}\label{co:SIZEio_SpanP_measure_0}
	For all $\alpha < 1$,
	$\SIZEio\left(\frac{2^n}{n}\left(1+\frac{\alpha \log n}{n}\right)\right)$
	has $\SpanP$-measure 0.
\end{corollary}
\begin{proof}
	This is immediate from Theorem \ref{th:SIZEio_MNP_measure_0} and Theorem \ref{th:NP-entropy-SpanP-Measure}.
\end{proof}

\begin{corollary}\label{co:lutzbound_DeltaEthree}
	For all $\alpha < 1$,
	$\SIZEio\left(\frac{2^n}{n}\left(1+\frac{\alpha \log n}{n}\right)\right)$
	has $\DeltaPthree$-measure 0 and measure 0 in $\DeltaEthree$.
\end{corollary}
\begin{proof}
	This is immediate from Corollary \ref{co:SIZEio_SpanP_measure_0} and Theorem \ref{th:counting_measure_conservation_deltaEthree}.
\end{proof}

Under a derandomization hypothesis, Corollary \ref{co:lutzbound_DeltaEthree} improves by one level in the exponential hierarchy to $\DeltaEtwo = \E^\NP$. This yields a stronger lower bound than other conditional approaches for obtaining lower bounds in $\E^\NP$ \cite{Aaronson:NECLB,Aydinlioglu:DAMG}.

\begin{corollary}\label{co:lutzbound_DeltaEtwo_derand}
	If Derandomization Hypothesis \ref{hypothesis:ENPtt_NPSVcircuits} is true, then  for all $\alpha < 1$,
	$\SIZEio\left(\frac{2^n}{n}\left(1+\frac{\alpha \log n}{n}\right)\right)$
	has $\DeltaPtwo$-measure 0 and measure 0 in $\DeltaEtwo$.
\end{corollary}
\begin{proof}
	This is immediate from Corollary \ref{co:SIZEio_SpanP_measure_0} and Theorem \ref{th:counting_measure_conservation_deltaEthree_derand}.
\end{proof}

Lutz \cite{Lutz:DCC} showed that
for all $\alpha \in [0,1]$, the class
\[ \calD_\alpha = \SIZE\left(\alpha\frac{2^n}{n}\right) \]
has $\pspace$-dimension $\alpha$. Athreya et al. \cite{Athreya:ESDAICC} showed that $\calD_\alpha$ also has strong $\pspace$-dimension $\alpha$.
Hitchcock and Vinodchandran \cite{Hitchcock:DERC} showed that
$\calH_\NP(\calD_\alpha) = \alpha,$
yielding the improvement that $\calD_\alpha$ has $\Deltapthree$-dimension $\alpha$.
It is immediate from this and Theorem \ref{th:CountingDimLeqEntropyRates} that $\calD_\alpha$ has $\SpanP$-dimension $\alpha$.
We improve this to $\SharpP$-dimension by using a Kolmogorov complexity argument.

\newcommand{\alphatwonn}{\alpha\frac{2^n}{n}}
\newcommand{\betatwonn}{\beta\frac{2^n}{n}}

\begin{theorem}\label{th:SIZE_dimSharpP}
	For all $\alpha \in [0,1]$,
	$\dimsharpp\left(\SIZE\left(\alpha\frac{2^n}{n}\right)\right) = \Dimsharpp\left(\SIZE\left(\alpha\frac{2^n}{n}\right)\right) = \alpha.$
\end{theorem}
\begin{appendixproof}[Proof of Theorem \ref{th:SIZE_dimSharpP}]
	Let $X_\alpha = \left(\SIZE\left(\alpha\frac{2^n}{n}\right)\right)$ and let $B \in X_\alpha$.
	For all sufficiently large $n$, $\pair{B_{=n},s} \in \MCSP$ for $s = \alpha \frac{2^n}{n}$. Let $C$ be a circuit of size $s$ on $n$ inputs. Frandsen and Miltersen \cite{FraMil05} showed that there exists a stack program of size at most $(s+1)(c+\log (n+s))$ that constructs $C$. Let $\gamma > \beta > \alpha$ be arbitrarily close rationals. Let $p(n) = n^3$ and  $q(n) = n^4$.
	For all $n$ larger than some $n_0$,
	\begin{eqnarray*}
		K^{p(2^n)}(B_{=n}) &\leq& (s+1)(c+\log (n+s)) + O(\log n) \\
		&\leq& \left(\alphatwonn +1\right) \left(c + \log \left(\alphatwonn + n\right)\right) + O(\log n)\\
		&\leq& \left(\betatwonn \right) \left(c + \log \left(\betatwonn\right)\right)\\
		&\leq& \left(\betatwonn \right) \left(c + \log \beta + n - \log n\right)\\
		&\leq& \left(\betatwonn \right) n\\
		&\leq& \beta 2^n.
	\end{eqnarray*}
	Let $N = 2^{n+1}-1$.  We have\begin{eqnarray*}
		K^{q(N)}(B_{\leq n})
		&\leq& \sum_{k=n_0}^n K^{p(2^k)}(B_{=k}) + O(n) \\
		&\leq& \beta N + O(n) \\
		&\leq & \gamma N,
	\end{eqnarray*}
	so \[\frac{K^{q(N)}(B \restr N)}{N} \leq \gamma\]
	for all sufficiently large $N$ of the form $2^{n+1}-1$. Therefore
	$\DimSharpP(X_\alpha) \leq  \gamma$ by Theorem \ref{th:Kolmogorov_rate_SharpP_dimension}.
	The dimension lower bound holds because $\dimSharpP(X_\alpha)\allowbreak\geq \dimh(X_\alpha) = \alpha$ \cite{Lutz:DCC}. The theorem follows because  $\alpha$ and $\gamma$ are arbitrarily close.
\end{appendixproof}

\begin{corollary}\label{co:Ppoly_SharpP_dimension}
	$\Ppoly$ has $\SharpP$-strong dimension 0.
\end{corollary}

Regarding infinitely-often classes, Hitchcock and Vinodchandran \cite{Hitchcock:DERC} showed that the class
$\SIZEiodimbound$ has $\NP$-entropy rate and $\Deltapthree$-dimension $\frac{1+\alpha}{2}.$
This extends to $\SharpP$-dimension, with a proof similar to Theorem \ref{th:SIZE_dimSharpP}.
\begin{theorem}\label{th:SIZEio_dimSharpP}
	For all $\alpha \in [0,1]$,
	$\dimsharpp\left(\SIZEiodimbound\right) = \frac{1+\alpha}{2}.$
\end{theorem}
\begin{appendixproof}[Proof of Theorem \ref{th:SIZE_dimSharpP}]
	Let $X_\alpha = \left(\SIZEiodimbound\right)$
	and let $B \in X_\alpha$.
	For all infinitely many $n$, $\pair{B_{=n},s} \in \MCSP$ for $s = \alpha \frac{2^n}{n}$.
	Let $\gamma > \beta > \alpha$ be arbitrarily close rationals, and let $p(n) = n^3$ and  $q(n) = n^4$.
	As in the proof of Theorem \ref{th:SIZE_dimSharpP},
	\begin{eqnarray*}
		K^{p(2^n)}(B_{=n})
		&\leq& \beta 2^n
	\end{eqnarray*}
	for infinitely many $n$.
	Let $N = 2^{n+1}-1$.  We have\begin{eqnarray*}
		K^{q(N)}(B_{\leq n})
		&\leq& K^p(B_{<n}) + K^p(B_{=n}) + O(n) \\
		&\leq& 2^{n}-1 + \beta 2^n + O(n) \\
		&\leq& \frac{1+\beta}{2} N + O(n) \\
		&\leq & \frac{1+\gamma}{2} N,
	\end{eqnarray*}
	so \[\frac{K^{q(N)}(B \restr N)}{N} \leq \frac{1+\gamma}{2}\]
	for infinitely many $N$ of the form $2^{n+1}-1$. Therefore
	$\DimSharpP(X_\alpha) \leq  \frac{1+\gamma}{2}$ by Theorem \ref{th:Kolmogorov_rate_SharpP_dimension}.
	The dimension lower bound holds because $\dimSharpP(X_\alpha)\allowbreak\geq \dimh(X_\alpha) = \frac{1+\alpha}{2}$ \cite{Hitchcock:DERC}. The theorem follows because  $\alpha$ and $\gamma$ are arbitrarily close.
\end{appendixproof}

We note that any infinitely-often defined class like $\SIZEiodimbound$ in Theorem \ref{th:SIZE_dimSharpP} always has its packing dimension and resource-bounded strong dimensions equal to 1 \cite{Gu:NDPSC}: $$\Dimsharpp\left(\SIZEiodimbound\right) = \dimpack\left(\SIZEiodimbound\right) = 1.$$

Li \cite{Li24}, building on work of Korten \cite{Korten22} and Chen, Hirahara, and Ren \cite{ChenHiraharaRen24}, showed that the symmetric exponential-time class $\SEtwo$  requires exponential-size circuits.
\begin{theorem_cite}{Li \cite{Li24}}
	$\SEtwo \not\subseteq \SIZEio\left(\frac{2^n}{n}\right)$.
\end{theorem_cite}
\noindent  Using  Lutz's counting argument \cite{Lutz:AEHNC} as in the proof of Theorem \ref{th:SIZEio_MNP_measure_0}, we improve this lower-bound to $\SIZEio\left(\lutzbound\right)$ for any  $\alpha < 1$.
\begin{theorem}\label{th:StwoE_lutzbound}
	For all $\alpha < 1$, $\SEtwo \not\subseteq \SIZEio\left(\lutzbound\right).$
\end{theorem}
\begin{appendixproof}[Proof of Theorem \ref{th:StwoE_lutzbound}]
	Li \cite{Li24} showed there is a single-valued $\FSPtwo$ function that given any polynomial-size circuit $C : \{0,1\}^n \to \{0,1\}^{n+1}$ outputs a nonimage of $C$. Li applies this algorithm to the truth-table generator circuit $\CCfont{TT}_{n,s}$ for $s = \frac{2^n}{n}$ and uses its $2^n$-length output to define as the characteristic string of the $\SEtwo$ language at length $n$. We observe that this proof works with $s = \lutzbound$ by Lutz's counting argument \cite{Lutz:AEHNC}.
\end{appendixproof}

\subsection{Quantum Circuit Complexity}

Recent work has also studied circuit-size complexity within quantum models. For instance, Basu and Parida \cite{BasuParida23} showed that the number of distinct Boolean functions on $n$ variables that can be computed by quantum circuits of size at most $c\frac{2^n}{n}$ is bounded by $2^{2^{n-1}}$, where $0 \leq c \leq 1$ is a constant that depends only on the maximum number of inputs of the gates. They proved this bound in a general setting in which the set of quantum gates is uncountably infinite.
Using universal gate sets with constant fan-in,  Chia et al. \cite{ChiaChouZhangZhang:arXiv21,ChiaChouZhangZhang:ITCS22} showed that the fraction of Boolean functions on $n$ variables that require quantum circuits of size at least $\frac{2^n}{(c+1)n}$ is at least $1 - 2^{-\frac{2^n}{c+1}}$.
We extend these quantum circuit-size bounds to a counting dimension result. In the following result, the $\BQSIZE$ class is independent of the choice of gate set because of the $o\left(\frac{2^n}{n}\right)$ size bound.
\begin{theorem}\label{th:BQSIZE_GapP_dimension_0}
	$\dim_\GapP\left(\BQSIZE\left(o\left(\frac{2^n}{n}\right)\right)\right) =
		\Dim_\GapP\left(\BQSIZE\left(o\left(\frac{2^n}{n}\right)\right)\right) = 0$.
\end{theorem}

\begin{appendixproof}[Proof of Theorem \ref{th:BQSIZE_GapP_dimension_0}]
	Let \(A\) be a language with quantum circuits of size $o(\frac{2^n}{n})$.
	We will use the Acceptance Probability Construction (Construction \ref{construction:probability_martingale_construction}) technique to construct a \(\GapP\)-martingale.
	Let $b \geq 1$ and let $s(n) = \frac{2^n}{bn}$.
	There exists \(n_0\) such that for all \(n \geq n_0\), the quantum circuit for \(A_{=n}\) has a size of at most \(s(n)\).
	Let \(\vec{C} = (C_{n_0},C_{n_0+1},\ldots)\) be these quantum circuits. We assume that the circuits are amplified so that \(C_{k}\) has error probability at most \(2^{-2k}\).

	Let \(d_{\vec{C}}\) be a martingale that on input \(x\)
	if \(2^{n_0}-1 \leq |x| \leq 2^{n+1}-1\) (i.e., we are betting on a string of length between \(n_0\) and \(n\)),
	runs the Acceptance Probability Construction with circuit \(C_{n_0}\) starting from length \(n_0\) on the bits of \(x\) corresponding to length \(n_0\). We bet on length \(n_0\) using \(C_{n_0}\), length \(n_0+1\) using \(C_{n_0+1}\), and so on, up to length \(n\) using \(C_{n}\).
	If \(|x| < 2^{n_0}-1\), \(d_{\vec{C}}(x) = 1\).
	Note that \(d_{\vec{C}}(\lambda) = 1\) and \(d_{\vec{C}}\) is a \(\GapP\) martingale.

	For any length \(k\) where \(n_0 \leq k \leq n\), \(d_{\vec{C}}\) wins \(\Omega(2^{2^k})\) by the analysis in Section \ref{sec:acceptance_probability_construction}.
	Therefore \(d_{\vec{C}}(A_{\leq n}) = \Omega(2^{2^{n+1}})\).
	Let \(\gamma > 0\) be a constant so that \(d_{\vec{C}}(A_{\leq n}) \geq \gamma 2^{2^{n+1}}\) for all sufficiently large \(n\).

	Define \(g_{n_0,n}\) on input \(x\in \{0,1\}^{\leq 2^{n+1}-1}\) to guess an extension \(w \in \{0,1\}^{2^{n+1}-1}\) with \(x \prefix w\) and guess a vector of quantum circuits \(\vec{C} = (C_{n_0},\ldots, C_n)\) where \(\mathrm{size}(C_i) \leq s(i)\) for all $i$ from  \(n_0\) to \(n\). Then \(g_{n_0,n}\) computes \(d_{\vec{C}}(w)\). Thus, \[g_{n_0,n}(x) = \sum\limits_{\vec{C}~:~(\forall i \in [n_0,n])\ \mathrm{size}(C_i) \leq s(i)} d_{\vec{C}}(x).\]

	By \cite{ChiaChouZhangZhang:ITCS22}, there is a constant $ c \geq 1$ depending on the universal gate set so that there are at most $2^{c s(n) \log s(n)} \leq 2^{\frac{c}{b(c+1)}2^n}$ quantum circuits of size \(s(n)\) for all sufficiently large \(n\).
	Let $\delta = \frac{c}{b(c+1)}$ and
	let $\epsilon > \delta$ be a dyadic rational.

	Define \(h_n(x)\) to be \(2^{\epsilon 2^{n+1}}\). Then
	\[d_{n_0,n}(x) = \frac{g_{n_0,n}(x)}{h_n(x)}\]
	is an exact \(\GapP\)-martingale. We have \[d_{n_0,n}(\lambda) \leq
		\frac{\prod_{m=n_0}^n 2^{\delta 2^m}}
		{2^{\epsilon 2^{n+1}}}
		=\frac{2^{\sum_{m=n_0}^n\delta 2^m}}{2^{\epsilon 2^{n+1}}}
		= 2^{\delta (2^{n+1}-2^{n_0})-\epsilon 2^{n+1}}
			= 2^{(\delta-\epsilon)2^{n+1}-\delta 2^{n_0}}.\]

	For \(m \geq 0\), define
	$$f_m(x) = \sum_{n_0=0}^m \sum_{n=n_0}^m d_{n_0,n}(x).$$
	This is a uniform family of exact \(\GapP\)-martingales.
	We have
	\[
		f_m(\lambda) = \sum_{n_0 = 0}^m \sum_{n = n_0}^m d_{n_0,n}(\lambda)
		\leq \sum_{n_0 = 0}^m \sum_{n = n_0}^m 2^{(\delta - \epsilon) 2^{n+1}}.
	\]
	There are \((m+1)^2\) terms in the double sum, each at most \(2^{(\delta - \epsilon) 2^{m+1}}\), so
	\[
		f_m(\lambda) \leq (m+1)^2 \cdot 2^{(\delta - \epsilon) 2^{m+1}}
		= 2^{(\delta - \epsilon) 2^{m+1} + 2 \log(m+1)}.
	\]
	This makes \(\sum\limits_{m=0}^\infty f_{m}(\lambda)\) \(\P\)-convergent because \(\epsilon > \delta\).
	Let
	$$f(x) = \sum_{m=0}^\infty f_m(x).$$
	By the Summation Lemma (Lemma \ref{lemma:counting-martingale-summation}), \(f\) is a \(\GapP\)-martingale.

	Let \(B \in \BQSIZE(o(\frac{2^n}{n}))\) be arbitrary with small quantum circuits starting at some \(n_0 \geq 0\). Let $\epsilon' > \epsilon$.
	We have
	\begin{eqnarray*}
		f(B_{\leq n})
		&\geq& f_n(B_{\leq n}) \\
		&\geq& d_{n_0,n}(B_{\leq n}) \\
		&\geq& \frac{\gamma 2^{2^{n+1}}}{2^{\epsilon 2^{n+1}}} \\
		&=& \gamma 2^{(1-\epsilon)2^{n+1}} \\
		&\geq& \gamma 2^{(1-\epsilon)(2^{n+1}-1)} \\
		&\geq& 2^{(1-\epsilon')(2^{n+1}-1)}
	\end{eqnarray*}
	for sufficiently large \(n\).
	Therefore \(f\) \(\epsilon'\)-strongly succeeds on \(B\).
	Since $B \in \BQSIZE(o(\frac{2^n}{n}))$ and $\epsilon' > \epsilon$ are arbitrary, \(\BQSIZE(o(\frac{2^n}{n}))\) has \(\GapP\)-strong dimension at most $\epsilon$. Since this holds for all $\epsilon > \delta$, the class has $\GapP$-strong dimension at most $\delta = \frac{c}{b(c+1)}$. Since $b \geq 1$ is arbitrary, the class has $\GapP$-strong dimension 0.
\end{appendixproof}

\begin{corollary}\label{co:BQPpoly_has_GapP_strong_dimension_0}
	$\BQP/\poly$ has $\GapP$-strong dimension 0.
\end{corollary}

Lutz \cite{Lutz:DCC} showed that $\SIZE(\alpha\twonn)$ has $\pspace$-dimension $\alpha$ and we improved this to $\SharpP$-dimension $\alpha$ in Theorem \ref{th:SIZE_dimSharpP}. It remains open whether a similar result holds for quantum circuits in either $\PSPACE$-dimension or $\GapP$-dimension (or even Hausdorff dimension). Achieving this would refine the current dimension zero statement into an exact dimension classification for small quantum circuits, but it would apparently depend on the choice of universal quantum gate set.

Another interesting direction is determining the measure or dimension of $\BQP/\qpoly$. While we know $\BQP/\poly$ has $\GapP$-dimension 0, extending this to quantum advice remains open. We note that Aaronson \cite{Aaronson:oracles_subtle} has shown that $\E^{\SharpP} \not\subseteq \BQP/\qpoly$, so it would be consistent with Theorem \ref{th:counting_measure_conservation} to show that $\BQP/\qpoly$ has $\GapP$-dimension 0.

A similar and simpler proof shows an infinitely-often version of Theorem \ref{th:BQSIZE_GapP_dimension_0}, analogous to Theorem \ref{th:SIZEio_dimSharpP}.

\begin{theorem}
	$\dimgapp(\BQSIZEio(o(\twonn))) = \frac{1}{2}.$
\end{theorem}

\subsection{Density of Hard Sets}

Investigations of the density of hard sets for complexity classes began with motivation from the Berman-Hartmanis Isomorphism Conjecture \cite{BerHar77}.
A language $S$ is {\em sparse} if $|S_{\leq n}| \leq p(n)$ for some polynomial $p$ and all $n$.
A language $S$ is {\em dense} if $|S_{\leq n}| > 2^{n^\epsilon}$ for some $\epsilon > 0$ and almost all $n$.
Let $\SPARSE$ be the class of all sparse languages and $\DENSE$ be the class of all dense languages.
Meyer \cite{Meyer77} showed that every hard set for $\E$ is dense.
A problem is in $\Ppoly$ if and only if it is in $\PT(\SPARSE)$, the polynomial-time Turing closure of $\SPARSE$ \cite{BerHar77,KarLip82}.
A line of subsequent work strengthened these results in multiple directions
\cite{Maha82,Wata87b,OgiWat91,BuhHom92,Lutz:MSDHL,Fu95,Lutz:DWCPAR,Hitchcock:OLRBD,Harkins:DHDHS,Buhrman:LRSS}.
The current best result for $\E$ is that the polynomial-time bounded-query Turing reduction closures
$\PnaT(\DENSE^c)$
and $\P_{o(n/\log n)-\CCfont{T}}(\SPARSE)$ both have $\p$-dimension 0, implying every hard language for $\E$ is dense or sparse, respectively, under these reductions \cite{Hitchcock:OLRBD}. Wilson \cite{Wilson85} showed that there is an oracle relative to which $\E $ has sparse hard sets under $O(n)$-truth-table reductions.
We now show that counting measure can handle polynomial-time Turing reductions to nondense sets, even if they are computed by $\Ppoly$ circuits.
Note that the following theorem extends Corollary \ref{co:Ppoly_SharpP_dimension}.
\begin{theorem}\label{th:PPolyT_density_SharpP_dimension_0}
	The class $(\Ppoly)_\T(\DENSE^c)$ of problems that $\Ppoly$-Turing reduce to nondense sets has $\SharpP$-dimension 0.
\end{theorem}
\begin{appendixproof}[Proof of Theorem \ref{th:PPolyT_density_SharpP_dimension_0}]
	Let $A \in (\Ppoly)_\T(\DENSE^c)$ be in the class.
	Composing the reduction with a lookup table for the nondense set shows that for all $\epsilon > 0$ and infinitely many $n$, the polynomial-time bounded Kolmogorov of $A_{\leq n}$ is at most $2^{n^\epsilon}$. We then apply Theorem \ref{th:Kolmogorov_rate_SharpP_dimension}.
\end{appendixproof}

\begin{corollary}
	Every problem that is $\Ppoly$-Turing hard for $\DeltaEthree$ is {\em dense}.
\end{corollary}
\begin{proof}
	This follows from
	Theorem \ref{th:PPolyT_density_SharpP_dimension_0},
	Proposition \ref{prop:basic_counting_measure_dimension_relationships}, and
	Corollary \ref{co:DeltaEthree_does_not_have_SpanP_measure_0}.
\end{proof}

\sectionnewpage
\section{Conclusion}\label{sec:conclusion}

We have introduced $\SharpP$, $\GapP$, and
$\SpanP$ counting measures and dimensions. These are intermediate in power between polynomial-time measure and dimension and polynomial-space measure and dimension. We have shown that counting measures and dimensions are useful for classes where the space-bounded measure or dimension is known but the time-bounded measure or dimension is not known.
This is the primary way to use counting measures and dimensions and we expect more results in this form.
\begin{enumerate}
	\item   If $\mupspace(X) = 0$ and $\mup(X)$ is unknown, investigate the counting measures $\musharpp(X)$, $\muspanp(X)$, and $\muGapP(X)$.
	\item If $\dimpspace(X) = \alpha$ is known and $\dimp(X)$ is  unknown, investigate the counting dimensions $\dimsharpp(X)$, $\dimspanp(X)$, and $\dimGapP(X)$.
	\item If $\Dimpspace(X) = \alpha$ is known and $\Dimp(X)$ is  unknown, investigate the counting strong dimensions $\Dimsharpp(X)$, $\Dimspanp(X)$, and $\DimGapP(X)$.
\end{enumerate}

We strengthened Lutz's $\pspace$-measure result by showing that the class of languages with circuit size $\frac{2^n}{n}(1+\frac{\alpha \log n}{n})$ has $\SpanP$-measure zero for all $\alpha < 1$. This improvement utilizes the Minimum Circuit Size Problem ($\MCSP$) to bridge the gap between $\pspace$-measure and $\SpanP$-measure. As a consequence, we showed that this measure-theoretic circuit size lower bound holds in the third level of the exponential-time hierarchy, $\DeltaEthree = \E^{\SigmaPtwo}$. Previously this was only known to hold in the exponential-space class $\ESPACE$. We also noted that recent work \cite{Li24} on exponential circuit lower bounds for the symmetric alternation class  $S^\E_2$ extends to this tighter $\frac{2^n}{n}(1+\frac{\alpha \log n}{n})$ bound. Under derandomization assumptions, our results further improve to the second level, $\DeltaEtwo = \E^\NP$. We showed that $\BQP$ and more generally, the class of problems with $o\left(\frac{2^n}{n}\right)$-size quantum circuits, has $\GapP$-strong dimension 0. This is the first work in resource-bounded measure or dimension to address quantum complexity.
Our results on circuit-size complexity are summarized in Figure \ref{figure:summary_of_results}.

We conclude with several open questions. The relationships between counting dimensions and other dimension notions, entropy rates, and Kolmogorov complexity rates are summarized in Figure \ref{fig:kolmogorov_relationships_diagram}.
\begin{question} Can any of the relationships in Figure \ref{fig:kolmogorov_relationships_diagram} be improved?
\end{question}

Arvind and K{\"o}bler \cite{ArvKob01} showed that for
each $\calC \in \{\parityP,\PP\}$, $\mup(\calC) \not= 0$ implies
$\PH \subseteq \calC$. Hitchcock \cite{Hitchcock:SSPP} showed that if $\SPP$ does not have $\p$-measure 0, then $\PH \subseteq\SPP$.
All of these classes have $\PSPACE$-measure 0 and their $\p$-measures are unknown, so the above paradigm applies:
\begin{question}
	What are the counting measures and dimensions of counting complexity classes including  $\PP$, $\ParityP$, and $\SPP$?
\end{question}
In particular, we know from Corollary \ref{co:SharpP_SpanP_random_UP_NP} that $\GapP$-random languages are not in $\SPP$. It is known that $\SPP$ is low for $\GapP$ \cite{FeFoKu94}, meaning that an $\SPP$ oracle provides no additional power to $\GapP$: $\GapP^\SPP = \GapP$. We do not know if $\SPP$ languages can be shown to uniformly have $\GapP$-measure 0.
\begin{question} Does $\SPP$ have $\GapP$-measure 0?
\end{question}
We showed in Theorem \ref{th:counting_measure_conservation_deltaEthree} that $\SpanP$-measure is dominated by $\DeltaPthree$-measure, which implies $\DeltaEthree$ does not have $\SpanP$-measure 0 or $\SharpP$-measure 0. For $\GapP$-measure, we only know that $\E^\GapP = \E^\SharpP$ does not have $\GapP$-measure 0.
\begin{question}
	What is the smallest complexity class that does not have $\GapP$-measure 0?
\end{question}

We also know that $\NP$ and $\UP$ have $\PSPACE$-measure 0, but we do not know their $\P$-measures. From Corollary \ref{co:SharpP_SpanP_random_UP_NP}, we know that $\SpanP$-random languages are not in $\NP$ and $\SharpP$-random languages are not in $\UP$. As with $\SPP$, it is not clear how to extend these proofs to show that the classes $\NP$ or $\UP$ have counting measure 0.
\begin{question}
	Does $\NP$ have $\SpanP$-measure 0?
\end{question}
\begin{question}
	Does $\UP$ have $\SharpP$-measure 0?
\end{question}

The {measure hypothesis} $\mup(\NP) \neq 0$ is known to have many plausible consequences \cite{Lutz:CVKL,Hitchcock:HHDCC,Lutz:TPRBM,Lutz:QSET}.
Hypotheses that complexity classes do not have counting measure 0 could be interesting to study. For example:
\[ \begin{array}{ccccc}
		\muspanp(\NP) \neq 0  & \Rightarrow &
		\musharpp(\NP) \neq 0 & \Rightarrow & \mup(\NP) \neq 0                                        \\
		\Downarrow
		                      &             & \Downarrow            &             & \Downarrow        \\
		\muspanp(\PP) \neq 0  & \Rightarrow & \musharpp(\PP) \neq 0 & \Rightarrow & \mup(\PP) \neq 0.
	\end{array} \]
Our results imply that if $\muSharpP(\NP) \neq 0$, then strong consequences hold for $\NP$:
\begin{enumerate}
	\item $\NP$ has problems with circuit-size complexity at least $\frac{2^n}{n}$ (Theorem \ref{th:SIZE_dimSharpP}).
	\item The $\leq^\Ppoly_\T$-hard languages for $\NP$ are dense (Theorem \ref{th:PPolyT_density_SharpP_dimension_0}).
\end{enumerate}
\begin{question}
	What else does $\mu_\SharpP(\NP) \neq 0$ imply? How does it compare with the standard measure hypothesis $\mup(\NP) \neq 0$?
\end{question}

We were able to show in Theorem \ref{th:SIZE_dimSharpP} that $\SIZE\left(\alpha\twonn\right)$ has $\SharpP$-dimension $\alpha$, but in Theorem \ref{th:BQSIZE_GapP_dimension_0} we only showed that $\BQSIZE\left(o\left(\twonn\right)\right)$ has $\GapP$-strong dimension 0.
\begin{question}
	Is it possible to determine the $\GapP$-dimension for the class of problems with $\alpha\twonn$-size quantum circuits, for an appropriate choice of universal gate set?
\end{question}
While $\BQP/\poly$ has $\GapP$-strong dimension 0, we do not know if this can be extended to quantum advice, even for $\GapP$-measure.
\begin{question}\label{question:BQPqpoly}
	Does $\BQP/\qpoly$ have $\GapP$-measure 0?
\end{question}
\noindent
Aaronson \cite{Aaronson:oracles_subtle} has shown that $\E^\SharpP \not\subseteq \BQP/\qpoly$, so
a positive answer to Question \ref{question:BQPqpoly} is consistent with Theorem \ref{th:counting_measure_conservation}.

\begin{acks}
	We thank Morgan Sinclaire and Saint Wesonga for helpful discussions.
\end{acks}

\sectionnewpage
\bibliographystyle{plainurl}
\bibliography{rbm,dim,random,main,newinclude,dimrelated,new,quantum}

\end{document}